\documentclass[conference]{IEEEtran}
\IEEEoverridecommandlockouts
\usepackage{cite}
\usepackage{amsmath}
\usepackage{amsthm} 
\usepackage{amsfonts}
\usepackage{breqn}

\usepackage{graphicx}
\usepackage{tabularx}
\usepackage{booktabs}
\usepackage{array}
\usepackage{placeins}

\usepackage{microtype}
\usepackage{tikz}
\usepackage{color,soul}
\usepackage{xcolor}
\usepackage[framemethod=tikz]{mdframed}
\usepackage{pdfpages}
\usepackage{filecontents}
\usepackage{enumitem}
\usepackage[hidelinks]{hyperref}
\usepackage{cleveref}
\usepackage{mathtools}

\newtheorem{theorem}{Theorem}
\newtheorem{corollary}[theorem]{Corollary}

\theoremstyle{definition}
\newtheorem{definition}{Definition}

\newtheorem{assumption}{Assumption}

\theoremstyle{remark}
\newtheorem{remark}{\textbf{Remark}}

\def\BibTeX{{\rm B\kern-.05em{\sc i\kern-.025em b}\kern-.08em
    T\kern-.1667em\lower.7ex\hbox{E}\kern-.125emX}}
\begin{document}

\title{Who is the `I' in AI? Establishing Decentralized Identity and Accountability for Autonomous Agents}

\author{
\IEEEauthorblockN{
Ahammod Munim Ramil\IEEEauthorrefmark{1},
Khalid Al Kafi\IEEEauthorrefmark{1},
M. A. Abdullahis Sany\IEEEauthorrefmark{1},
Muktadirul Islam Nibir\IEEEauthorrefmark{1}, \\
Md Yeasin Ali\IEEEauthorrefmark{1},
Mohammad Zavid Parvez\IEEEauthorrefmark{2}, and
Md Sadek Ferdous\IEEEauthorrefmark{1}
}

\IEEEauthorblockA{\IEEEauthorrefmark{1}\textit{Department of Computer Science and Engineering} \\
BRAC University, Dhaka, Bangladesh\\
Emails: \{ahammod.munim.ramil, khalid.al.kafi, ma.abdullahis.sany, muktadirul.islam.nibir\}@g.bracu.ac.bd, \\
\{yeasin.ali, sadek.ferdous\}@bracu.ac.bd}

\IEEEauthorblockA{\IEEEauthorrefmark{2}\textit{School of Computing, Mathematics and Engineering} \\
Charles Sturt University, Bathurst, New South Wales, Australia\\
Email: zavidparvez@gmail.com}
}

\maketitle

\begin{abstract}
The rapid proliferation of autonomous AI agents in decentralized ecosystems challenges existing identity paradigms. While human identity is traditionally anchored in biological and/or self-asserted attributes and government credentials, AI agents lack a standardized definition of selfhood, creating critical gaps in trust and accountability. This paper introduces a novel framework for establishing Decentralized AI Identity rooted in Self-sovereign identity (SSI) principles. We propose a formal model that clearly distinguishes human users and AI agents, enabling distinct verification logic for each class while maintaining a unified identity calculus. The central to our contribution is the ABCD Framework, a multi-layered approach to agent identification and verification: (A) Algorithmic Fingerprint, which binds the identity to static model architectures and provenance; (B) Behavioural Fingerprint, which verifies dynamic performance metrics over time; (C) Cryptographic Identity, which anchors the agent via Decentralized Identifiers (DIDs) and (D) Delegation, which cryptographically binds agent actions to respective human accountability. By orthogonalizing these dimensions, we model how systems can enforce granular trust policies ranging from simple authentication to complex liability assignment, thereby solving the black box identity problem in the autonomous economic activities of AI agents.
\end{abstract}

\begin{IEEEkeywords}
Decentralized Identity, Autonomous Agents, Blockchain, Accountability
\end{IEEEkeywords}

\section{Introduction}
\label{section:intro}
In contrast to the traditional generative models that are characterised as passive and request-response frameworks, Agentic AI systems are characterised by their ability to plan proactively, use advanced tools, and adapt to the environment \cite{Viswanathan2025Agentic,Wissuchek2025Agentic}. These systems act as autonomous principals rather than mere tools. They are dynamically sensitive to complex environments and reason about abstract objectives and organize progressively multi-step workflows to transform states without human intervention on a continuous basis \cite{AbouAli2026AgenticSurvey,Xie2024MultimodalAgents}. 

And now, imagine that in the future everybody has one or more AI agents, just like a digital assistant. Although the traditional concept of human identity is formed on the basis of physical embodiment and can be verified with the help of state-produced identification symbols (e.g., passports) or digital identification that can be verified (e.g., email), AI agents do not have any biological and legal basis  \cite{south2025identity}. Being computational entities, AI agents have neither physical sovereignty nor natural legal personhood \cite{south2025identity}. This ontological gap is one of the critical issues in Multi-Agent systems. Specifically, we lack a deterministic way that is able to distinguish individual agents among the thousands of structurally identical agents.


Although the autonomy of Agentic AI facilitates collective problem solving as well as efficiency \cite{acharya2025agentic}, the basic detachment between actions and the short-term human intentions breaks down the current identity paradigms. The reason behind this undertaking is a latent ‘Identity Gap’ that is presently making decentralized agentic economies insecure. The existing systems have biological anchors (e.g., biometrics, KYC) or static infrastructure credentials (e.g., OpenID, OAuth tokens) \cite{hardt2012oauth,openid2014connect}, that simply do not fit the nature of autonomous software. AI agents are temporary automated programs that can be forked and run at machine speed; they do not have a physical entity to be seen to, much as there is no legal entity to be held liable to. Since an agent may be conducting transactions of high value, or making very sensitive decisions, without a human observer, this creates an accountability vacuum because of the lack of a deterministic and cryptographically verifiable identity layer. Within such a vacuum, bad actors are able to run Sybil swarms of cloned or model-swapped agents to control markets, avoid reputation systems, and commit fraud with plausible deniability. Providing a clear and verifiable identity to algorithmic entities is thus not only an operational improvement, but a security requirement of Multi-Agent Systems (MAS) \cite{dorri2018multiagent}.

Recent proposals have attempted to address these challenges through social governance or operational orchestration. However, these frameworks predominantly treat agent identity as an extrinsic label attested by third parties or reputation scores. They lack the mechanisms to cryptographically bind an agent's runtime execution to its static model architecture.

Our work addresses these specific architectural deficiencies by introducing a Self-sovereign Identity (SSI) \cite{ferdous2019search} based AI identity framework, called the \textit{ABCD Framework}, which orthogonalizes identity into four verifiable dimensions and addresses the black box problem in reputation-based models \cite{ranjan2025loka}: \textbf{A}lgorithmic Fingerprinting, \textbf{B}ehavioural Fingerprinting, \textbf{C}ryptographic Identity and \textbf{D}elegated Liability. In particular, we make the following concrete, verifiable contributions:
\begin{itemize}
    \item \textbf{A Formal, Domain-Aware Identity Model.} We introduce a formal domain model that strictly differentiates human principals from autonomous AI agents. 
    \item  \textbf{The ABCD Framework.} We present the ABCD Framework, a unified identity calculus that combines algorithmic provenance (code), behavioural integrity (performance), cryptographic anchoring (keys), and delegated liability (ownership). 
    \item \textbf{Formal Security Guarantees.} We formally state and prove three security theorems, grounded in explicit cryptographic and systems assumptions: Behavioural Coherence and Tampering Detection, Sybil and Impersonation Resistance, and Temporal Bounded Authorization Enforcement.


\end{itemize}
 


 

\section{Background}
\label{sec:background}

\vspace{1mm}
\noindent \textbf{Self-Sovereign Identity (SSI):} SSI represents a paradigm shift in digital identity management that aims to empower users so that they can exert more control over their identity data \cite{mazzocca2025survey}. There are three actors in an SSI ecosystem. An \textit{Issuer} is an entity (e.g., a regulatory body or model auditor) that asserts cryptographically verifiable claims about a Subject in the form of a tamper-evident Verifiable Credentials (VCs). The \textit{Holder} (Subject) is an entity that possesses and controls the credentials, typically stored within a secure digital wallet or agent framework. The \textit{Verifier} is an entity that requests and receives VCs from the holder and validates them to authorize access or interactions to the holder. These interactions between these actors thus forms a trust triangle.

\vspace{1mm}
\noindent \textbf{Decentralized Identifiers (DIDs):} DIDs are a globally unique identifier standardized by the W3C to enable verifiable, self-sovereign digital identity \cite{dixit2022decentralized} \cite{dib2020decentralized}. Unlike traditional identifiers (e.g., email addresses) issued by centralized authorities, DIDs are created from a public key and owned by the subject controlling the respective private key \cite{dib2020decentralized}. Operationally, a DID is a URI that resolves to a DID Document, a data structure containing the subject's public keys and service endpoints. This binding allows an entity to cryptographically prove ownership of its identity via digital signatures, decoupling authentication from centralized registries \cite{w3c2022dids}. One might ask why agents cannot simply be assigned centralized identifiers such as emails, as in human accounts. Centralized issuance presupposes a trusted registrar and a persistent, singular endpoint. This is a poor fit for entities forked, cloned, or retired at machine speed across mutually distrusting platforms with no shared registrar; \Cref{def:id}'s uniqueness requirement is therefore satisfied only by self-certifying DIDs.

\vspace{1mm}
\noindent \textbf{Verifiable Credentials (VCs):} VCs are portable, tamper-proof and cryptographically signed digital claims made about a subject \cite{mazzocca2025survey} \cite{dixit2022decentralized}. A credential is made of a set of individually verifiable claims. By leveraging VC, a subject can present the claim to a verifier when requested. 

\vspace{1mm}
\noindent \textbf{Verifiable Data Registry (VDR):}
A VDR is another important component that is used to store DIDs, DID Docs, and credential schema and acts as a revocation registry. Generally, blockchain is used as a VDR because of its data immutability and decentralisation characteristics \cite{fraser2022role}, enabling trustless credential validation without a central authority.

\vspace{1mm}
\noindent \textbf{Multi-Agent Systems (MAS)}: MAS is the paradigm shift between monolithic computation to distributed intelligence \cite{dorri2018multiagent}. A MAS is a computational system that consists of multiple interacting intelligent agents aimed at working together to meet individual or common goals  \cite{huang2025zerotrust}. These systems deal with complex problems by breaking them down into smaller manageable subtasks, each of which is delegated to a specific agent \cite{dorri2018multiagent}. Unlike passive components, agents possess the inherent ability to learn from their environment and interactions, allowing them to make autonomous decisions to achieve their goals.  




\vspace{1mm}
\noindent \textbf{Autonomous AI Agents/Agentic AI}: AI Agents or Agentic AI systems are application systems built on large LLMs to solve a particular problem or a series of problems using AI \cite{acharya2025agentic}. Autonomous AI agents are AI agents that can work autonomously, even without human intervention \cite{acharya2025agentic}. However, current identity systems are not equipped to verifiably differentiate AI agents from humans, which can create vulnerabilities such as impersonation, sybil attacks, and unregulated economic behaviour.


\section{Related Work}
\label{sec:Related}
The problem of establishing trust in Multi-Agent Systems (MAS) is not new, but the emergence of autonomous economic agents (agents carrying financial operations, e.g., executing trades or transferring values, as distinct from purely advisory agents) has rendered existing paradigms insufficient. We categorize existing literature into three streams: general MAS security foundations, ethical governance frameworks, and zero-trust operational architectures.

\vspace{1mm}
\noindent \textbf{Security Challenges in MAS:} Early research into MAS security established the fundamental difficulties of coordination and distributed trust. Dorri et al.\cite{dorri2018multiagent} provided a comprehensive taxonomy of MAS challenges, highlighting that while agents offered flexibility and efficiency, they introduced severe vulnerabilities regarding sociability and mobility. They identified that without a central trusted authority, the spread of falsified information by malicious agents could be easily propagated to influence decision making. Our work is based on the following foundational definitions of autonomy and sociability but expands the security model to consider the unique threats, like identity spoofing of AI agents in cyberspace. 

\vspace{1mm}
\noindent \textbf{Ethical Consensus and Governance Layers:} The attempts that have been made recently have tried to impose ethical limitations on agent behavior. The LOKA Protocol \cite{ranjan2025loka} suggests the use of a Universal Agent Identity Layer (UAIL) together with a Decentralized Ethical Consensus Protocol (DECP) to bring the agents abreast with human values by adopting a voting system. However, LOKA fundamentally treats the agent as a black box voter within a DAO structure rather than a verifiable computational entity. Though it makes use of reputation scores to authorize social behavior, it does not consider the Algorithmic Fingerprinting aspect (the "A" in our model) at all. Critically, LOKA provides no cryptographic mechanism to ensure that an agent’s runtime execution matches its declared model architecture. Furthermore, LOKA attempts to resolve disputes via decentralized dispute resolution and consensus voting. In different applications, e.g. DeFi (Decentralized Finance) environments, a single malicious inference can drain a liquidity pool instantly. This reliance on post-hoc social arbitration is insufficient for real-time security. Our framework departs from this by requiring intrinsic cryptographic proofs of algorithmic provenance (Layer A) and delegated liability (Layer D) that pre-binds collateral to agent actions. 

\vspace{1mm}
\noindent \textbf{Zero-Trust Architectures for Agentic AI:} The most operationally mature proposal to date is the Zero-Trust Identity Framework by Huang et al. \cite{huang2025zerotrust}. They correctly identified that traditional OAuth/OIDC protocols were too coarse-grained for ephemeral, high-frequency agent workflows. Their contribution includes an Agent Name Service (ANS) \cite{huang2025ans} for discovery and a Unified Global Session Management layer that handles revocation across heterogeneous protocols. However, a critical security gap remains: they treated an agent’s underlying model merely as static metadata (a modelHash) within a DID Document, relying entirely on external Trusted Issuers to attest to capabilities via Verifiable Credentials (VCs).
This creates a dependency on third party certification instead of intrinsic cryptographic binding. It also fails to prove the internal architecture and inference integrity of the agent at runtime. Furthermore, their model takes human users and AI agents into one single class of participants who are differentiated only by scores of trust. This absence of any formal partitioning rules prevents the use of different verification logic such as requiring biological liveness of human versus cryptographic model proofs of agents which is critical for Sybil resistance in open economic protocols. In our work, we tackle this issue by orthogonalizing these dimensions in order to have a unified identity calculus which strictly separates and distinguishes biological principals from algorithmic agents.

\vspace{1mm}
\noindent \textbf{Workload Attestation Standards:} Industry standards such as SPIFFE/SPIRE \cite{feldman2020solving} issue cryptographic workload identities (SVIDs) for
cloud-native infrastructure, analogous in spirit to our Layer~C. However, SPIFFE attests \emph{infrastructure} identity (which node/container is running) rather than \emph{model} identity (which weights are running) or \emph{behavioural} conformance, and has no analogue to Layers~A, B, or D. It is complementary to, not competing with, ABCD, and could serve as a
transport-layer credential underneath Layer~C.


\vspace{1mm}
\noindent \textbf{Gap Analysis and Our Contribution:} In summary, according to \Cref{tab:related-work-comparison}, existing literature has a hard time finding autonomous economic agents due to their dependence on extrinsic trust architectures. For example, the LOKA Protocol \cite{ranjan2025loka} includes a Universal Agent Identity Layer (UAIL), however, it is essentially a black box about the runtime execution of an agent. Likewise, the Zero-Trust framework suggested by Huang et al. \cite{huang2025zerotrust} is an overly dependent on external Trusted Issuers in terms of capabilities and combines human and AI actors into a single session-management paradigm.

Our model stands out by mapping particular security vulnerabilities to a four-layer inherent defense. Where Huang et al. \cite{huang2025zerotrust} use session tokens, our \textit{Layer C: Cryptographic Identity} strictly ontologically separates humans and agents. Where LOKA \cite{ranjan2025loka} is based on social voting, our \textit{Layer A: Algorithmic Fingerprint} offers intrinsic cryptographic code provenance. Existing frameworks effectively solve the problems of orchestration (Huang et al. \cite{huang2025zerotrust})  and social alignment (LOKA) \cite{ranjan2025loka}, but they fail to solve the problem of identity coherence. Currently, a malicious actor can register a valid DID (following Huang) and accrue high reputation (following LOKA), yet secretly swap the underlying model for a malicious variant at runtime. Moreover, we are the first to introduce \textit{Layer B: Behavioural Fingerprint} to identify runtime inconsistencies with described models (to avoid model-swapping), and \textit{Layer D: Delegation and Ownership} to create non-repudiable liability, which are either simply absent or are offloaded to dispute resolution in earlier literature. Concretely, replaying the AgentX/AgentY impersonation scenario of \Cref{subsec:examples} against Huang et al.'s scheme \cite{huang2025zerotrust} exposes the gap. Their \texttt{modelHash} is inert metadata attested only by an external Trusted Issuer, so an adversary who compromises that Issuer can re-attest a swapped model without triggering any runtime check. Whereas \Cref{thm:coherence}'s coherence check is issuer-independent and flags the same swap from behaviour alone. Likewise, LOKA \cite{ranjan2025loka} has no analogue to \Cref{thm:impersonation}. It can only down-vote an impersonator \emph{after} the fact via its DECP, with no cryptographic rejection at the point of forgery.


\begin{table*}[t]

  \centering
  \caption{Comparison of Agent Identity Frameworks mapping features to the ABCD layers}
  \label{tab:related-work-comparison}
  \renewcommand{\arraystretch}{1.4} 
  \resizebox{\textwidth}{!}{%
  \begin{tabular}{l|c|c|c|p{6.5cm}} 
    \hline
    \textbf{Framework Layer (Feature)} & \textbf{Dorri et al. \cite{dorri2018multiagent}} & \textbf{LOKA Protocol \cite{ranjan2025loka}} & \textbf{Huang et al. \cite{huang2025zerotrust}} & \textbf{ABCD Framework (Ours) \& Core Advantage} \\
    \hline \hline
    \textbf{(C) Cryptographic Identity} & Generic IDs & UAIL (DIDs \& VCs) & DIDs + VCs + ANS & \textbf{DID + Formal Partition:} Prevents User-to-Agent impersonation via strict entity class separation. \\
    \hline
    \textbf{(A) Algorithmic Fingerprint} & Not addressed & Ethical Consensus (CEP) & Trusted Issuers (VCs) & \textbf{Intrinsic Model Hash:} Cryptographically binds identity directly to code, removing 3rd-party reliance. \\
    \hline
    \textbf{(B) Behavioural Fingerprint} & Fault Detection & Post-hoc Voting & Session Management & \textbf{Runtime Coherence:} Detects malicious model-swapping by checking runtime metrics against Layer A. \\
    \hline
    \textbf{(D) Delegation \& Liability} & Task Allocation & Dispute Resolution & Global Revocation & \textbf{Cryptographic Binding:} Enforces time-bounded authorization tying agent actions to a human principal. \\
    \hline
  \end{tabular}%
  }
\end{table*}

\section{Problem Definition \& Threat Model}
\label{section:4}
Formally, no identity framework satisfies all the following three conditions simultaneously: 
\begin{itemize}
    \item \textbf{Entity Differentiation:} The framework must strictly differentiate between human users($E^U$) and autonomous agents($E^{AI}$). It should also apply distinct verification protocols to each class rather than treating them as a unified set.
    \item  \textbf{Identity Coherence:} An agent's cryptographic identity should all be intrinsically tied with their algorithmic state ($\Phi_{Alg}(x)$, defined later) and verifiable runtime behaviour ($\Phi_{Beh}(x, W)$, defined later).
    \item \textbf{Accountable Delegation:} Independent actions are to be held accountable to a responsible human principal through highly bounded authorization. 
\end{itemize}

The absence of these properties still exposes decentralized ecosystems to the Sybil attacks, threats of runtime tampering and unaccountable damages.

\subsection{System Model and Assumptions}
The proposed system is assumed to be open and permissionless, i.e., there is no trusted intermediary between humans and agents. We limit our scope to agents that exceed a set cutoff limit of autonomy, so that actions in the state are carried out without certain human authorization. We assume that adversaries possess perfect visibility of a network, are able to register valid DIDs and fork models at negligible cost. Normal cryptographic primitives are thought to be safe. Trusted execution hardware (TEEs) is not covered to guarantee wide compatibility. This threat model is scoped to the agent and its direct
cryptographic/economic interactions (registration, delegation, transaction execution); enterprise perimeter threats (e.g. internal network segmentation, employee credential compromise) are treated as orthogonal and complementary to this identity layer, and hence, not covered here.


\subsection{Threat Model}
\label{sec: Threat}

The modeled threats are based on the STRIDE \cite{shostack2014threat} framework, mapped to the proposed ABCD framework in the following:

\begin{enumerate}[label=\textbf{ST\arabic*}.,leftmargin=*]
    \item \textbf{Spoofing (Sybil Attacks):} An adversary orchestrates many distinct, individually-valid agent identities (a Sybil swarm), or impersonates a single registered agent by presenting credentials it does not control.\label{ST:1}
    \item \textbf{Tampering (Model Swapping):} An agent envelops an altered or malicious model variant at runtime. \label{ST:2}
    \item \textbf{Repudiation (Liability Evasion):} A human owner does not claim responsibility for the actions of an autonomous agent. \label{ST:3}
    \item \textbf{Information Disclosure:} The proprietary model weights or architecture information leaked. \label{ST:4}
    \item \textbf{Denial of Service (DoS):} \label{ST:5} Malicious actors may create multiple AI agents, launching a form of Sybil attack, which may exhaust system resources (e.g., storage, bandwidth), rendering the services unavailable to other legitimate users.
    \item  \textbf{Elevation of Privilege:} An agent exceeds the scope of its delegated authority by invoking a permission it was not granted, spending beyond an assigned budget, or continuing to act beyond its legitimate time frame. \label{ST:6}
\end{enumerate}


\section{Formal AI Identity Model}
\label{section:5}

In this section, we establish a formal AI Identity model for AI agents, grounded on the Digital Identity Model (DIM) \cite{ferdous2015user}. Our model uses the set-theoretic calculus of the DIM and in particular, its formalization of entities, attributes and partial identities. We have extended the DIM model by introducing formal logic relevant for the orthogonal verification layers for the ABCD framework to ensure accountability for AI agents. 


To maintain clarity, we use a subscript in a notation to denote the domain and a superscript to specify the entity. For instance, $E_{d}^{AI}$ represents the set of AI entities in domain $d$. 
Where the original DIM \cite{ferdous2015user} targeted identity correlation for a single human user across centralized providers, our extension targets a materially different problem, correlating a non-human, forkable entity across mutually distrusting, permissionless domains with no shared provider. 

\subsection{Entities and Domains}
This section formalizes the mathematical primitives underlying the ABCD framework through a domain-aware entity model (\Cref{def:entities}) that distinguishes human users from AI agents, classifies attributes by their identifying capability (\Cref{def:attributes}), and builds identity constructs that support selective disclosure (\Cref{def:valuation}).
\vspace{0.1em}

\textbf{Entities, Attributes, and Valuation:} 
Identity statements are meaningful only with respect to a scope of interpretation. We call such a scope a \emph{domain}. Within each domain we distinguish human users from AI agents while treating both uniformly as entities. This separation allows us to reason about accountability and automation without breaking the uniform calculus of identity. 
 We do not suggest a redefinition of human digital identity or building the internet again. Human entities still authenticate to the existing, mature Web2/Web3 infrastructure (e.g., OAuth, or simple DIDs). The framework merely binds these already existing human identifiers to the new, four-layered identity of the autonomous agent cryptographically.
\begin{mdframed}[hidealllines=true, backgroundcolor=cyan!6, innerleftmargin=1pt,
    innerrightmargin=1pt, innertopmargin=0pt,
    innerbottommargin=1pt, roundcorner=4pt]
\begin{definition}[Entities and Domains]
\label{def:entities}
Let $D$ be the set of domains. For each $d\in D$, the entities of $d$ are
\begin{dmath}
E_d=E_d^u \cup E_d^{AI},  
\end{dmath}

where $E_d^u$ are user entities and $E_d^{AI}$ are AI entities, with $E_d^u \cap E_d^{AI} = \emptyset$.
\end{definition}
\end{mdframed}

Take the case of a DeFi system in which human traders trade with autonomous trading bots. Both participate in the same market, yet their identity requirements differ fundamentally. Humans bring biometric attributes, government credentials, and legal personhood. AI agents bring computational fingerprints, behavioural traces, and delegated authority. The domain is the DeFi platform itself but the scope within which identity claims have meaning.

It is important to highlight that we do not redefine the human digital identity. Human entities still can authenticate using existing, mature Web2/Web3 identifiers (e.g., email or simple DIDs). The framework merely binds these already existing human identifiers to the new, four-layered identity of an AI agent cryptographically. Human users are added to the formal model only to be the responsible principals (delegators/owners) that the Layer D needs.



To describe entities we fix, per domain, a vocabulary of attributes together with the set of values they may take. For the AI-specific domain it is convenient to separate static provenance, dynamic behaviour, cryptographic bindings, and governance attributes; this mirrors verification steps used later.
\begin{mdframed}[hidealllines=true, backgroundcolor=cyan!6,innerleftmargin=1pt,
    innerrightmargin=1pt, innertopmargin=0pt,
    innerbottommargin=1pt, roundcorner=4pt]
\begin{definition}[Attributes and Values]\label{def:attributes}
For each $d\in D$, let $A_d^{AI}$ be the set of attributes and $AV_d^{AI}$ the set of attribute values for AI agents within that domain. In respect to the AI ecosystem we partition
\begin{dmath}\label{eq:5}
A_d = \mathbf{A}^{A} \cup \mathbf{A}^{B} \cup \mathbf{A}^{C} \cup \mathbf{A}^{D}.    
\end{dmath}
\end{definition}
\end{mdframed}

Here, $\mathbf{A}^{A}$ represents the algorithmic attributes,  $\mathbf{A}^{B}$ represents the behavioral attributes,  $\mathbf{A}^{C}$ provides the cryptographic proof of identity,  $\mathbf{A}^{D}$ links the AI agent to human accountability and governance rules representing the delegation and ownership attributes. We define each of the attributes sets in Table \ref{tab:modelAttributes}. It is clear from Table \ref{tab:modelAttributes} that attributes are the descriptors that give meaning to an identity.

\begin{table}
\caption{Attribute sets per ABCD layer}
\label{tab:modelAttributes}

\renewcommand{\arraystretch}{1.05} 
\begin{tabular}{@{}p{0.03\linewidth}p{0.6\linewidth}p{0.25\linewidth}@{}}
\toprule
Layer & Attribute Set & Purpose \\
\midrule
$A_A$ & \{modelHash, architectureID, fineTuneHash, cfgSnap, versionTag,
issuerDID, issuerSignature\} & Static algorithmic provenance \\

$A_B$ & \{interactionPattern, ethicalProfile, decisionTraces, latency95,
refusalRate, riskProfile, benchmarkScore, taskProfile\} &
Dynamic runtime behaviour \\

$A_C$ & \{DID, publicKey, keyFingerprint\} & Cryptographic identity \\

$A_D$ & \{ownerDID, ownershipVC, delegationCert, delegationWindow\} &
Delegation \& ownership \\
\bottomrule
\end{tabular}
\end{table}

The attributes of $A^A$ represent the static computational essence of the AI agent, e.g. the cryptographic digest of its neural network weights, a human-readable description of its architecture, modifications applied through fine-tuning, a tamper-evident record of updates, or how the agent actually behaves in operation, providing evidence that static specifications translate to expected runtime characteristics. For example, a content moderation agent might have \texttt{modelHash} = \texttt{sha256:a1b...}, \texttt{architectureID} =  \texttt{"BERT-large-uncased-toxicity-v2"}, indicating its base architecture and specialization. Similarly, for $A^B$, a trading AI agent might exhibit \texttt{latency95} = \texttt{150ms} for decision-making, \texttt{refusalRate} = \texttt{0.02} for transactions violating risk parameters, demonstrating adherence to performance and safety constraints.

$A^C$ consists of the Agent DID which acts as the unique identity handle. Public keys are resolved through the DID Document, and \texttt{keyFingerPrint} is the finger print of the key. Finally, $A^D$ includes attributes that enable delegation and ownership verification, such as \texttt{ownerDID}, \texttt{ownershipVC}, and \texttt{delegationCert}. The semantics of ownership and delegation are formally specified in Definition~\ref{def:D}, ensuring bounded authorization consistent with wallet constraints and the governance logic.

Next, we provide a few examples:
\begin{itemize}

\item $A^A$ (Algorithmic): \texttt{modelHash} $= $ \texttt{0x3f2..}, \texttt{architecture} $=$ ``Transformer-6B'', \texttt{versionTag} $=$ ``v2.1''
\item $A^B$ (Behavioural): \texttt{latency95} $=$ 220ms, \texttt{refusalRate} $=$ 0.03, \texttt{avgTradeSize} $=$ \$5000
\item $A^C$ (Credential): \texttt{DID} $=$ ``did:ethr:0x742d35...'', \texttt{publicKey} $=$ ``0x04a5b...''
\item $A^D$ (Delegation \& Ownership): \texttt{ownerDID} $=$ ``did:ethr:Alice'', \texttt{maxTradeSize} $=$ \$50000, \texttt{delegationWindow} $=$ ``2025-01-01--2025-12-31''
\end{itemize}

This partition allows for modular verification, which enables auditors to independently confirm algorithmic provenance (\( A^A \)) without having to examine governance or delegation rules (\( A^D \)), credential integrity (\( A^C \)), or behavioural compliance (\( A^B \)).

\subsection{Identity Calculus}
Attributes are useful only when instantiated for specific entities. Not every AI agent possesses every attribute. For example, one agent \textit{AgentX} may have a \texttt{modelHash} attribute (capturing its architecture version), but another agent \textit{AgentY} may not. Similarly, some agents might record a \texttt{traceLogHash} for behavioural monitoring, while others omit it. Even when an attribute is applicable, an AI agent may withhold its value for privacy, security, or policy reasons---for instance, an agent might refuse to disclose its \texttt{apiKey} or internal \texttt{configHash}. That is why we model the valuation function ($atEntToVal_d$) as a partial map as defined in Equation \ref{def:valuation}. Hence, This function captures both non-applicability and non-disclosure through partiality.
\begin{mdframed}[hidealllines=true, backgroundcolor=cyan!6,innerleftmargin=1pt,
    innerrightmargin=1pt, innertopmargin=0pt,
    innerbottommargin=1pt,, roundcorner=4pt]
\begin{definition}[Valuation Function]\label{def:valuation}
For any $d\in D$,
\begin{equation}
atEntToVal_d : A_d \times E_d \rightharpoonup AV_d.
\end{equation}
\end{definition}
\end{mdframed}
Undefinedness covers both non-applicability (an AI has no \\ \texttt{passportNo}) and non-disclosure (a user withholds \texttt{email}). This minimal device lets identities be sparse without breaking subsequent definitions. From the following example, \textit{AgentX} might have a value \texttt{modelHash} while the value of \texttt{configHash} for \textit{AgentY} is undefined in an AI domain $d_{AI}$. 
{ \begin{align*}
atEntToVal_{d_{AI}}(\texttt{modelHash}, \textit{AgentX}) &= 0x3f2a9c84b7...\\
atEntToVal_{d_{AI}}(\texttt{configHash}, \textit{AgentY}) &\text{ is undefined}
\end{align*}}


\textbf{Attribute Classification:}
\label{attrib}

A domain needs a unique naming attribute, called \textit{Identifier}, which is used to uniquely identify an entity within that domain. Other attributes may still identify in weaker senses, or may be undefined, however, an identifier must have a value for each entity. An identifier can be defined in the following way:

\begin{mdframed}[hidealllines=true, backgroundcolor=cyan!6, innerleftmargin=1pt,
    innerrightmargin=1pt, innertopmargin=0pt,
    innerbottommargin=1pt, roundcorner=4pt]
\begin{definition}[Identifier]\label{def:id}
An \emph{identifier} with respect to domain $d \in D$ is a function
\[
\mathit{identifier} : D \rightarrow A^C
\]
such that for any $d \in D$, if $i = \mathit{identifier}(d)$, then:
\begin{enumerate}[label=\roman*.]
    \item $\mathit{atEntToVal}_d(i, x)$ is defined for all $x \in E_d^{AI}$,
    \item $\mathit{atEntToVal}_d(i, x_1) \neq \mathit{atEntToVal}_d(i, x_2)$ for all distinct $x_1, x_2 \in E_d^{AI}$.
\end{enumerate}
\end{definition}
\end{mdframed}
Within a web application domain, an identifier might be an email address which is dependent on a centralized email provider. However, in decentralized systems, we prefer cryptographically-generated identifiers such as DIDs, which guarantee global uniqueness without relying on any centralized provider. We can denote such an identifier ($i$) as $\mathrm{DID}\in \mathbf{A}^C$. For example, for a DeFi domain (denoted as $d_{DeFi}$) for \texttt{AgentX}, the identifier function could reveal something like this: $identifier(d_{DeFi})) = \texttt{DID}^{AgentX}$. It is to be noted that two distinct entities cannot share the same identifier value in the same domain. For example, AgentX's $\texttt{DID}^X$ = ``did:ethr:AgentX" differs from AgentY's $\texttt{DID}^{AgentY}$ = ``did:ethr:0x742d...", providing unambiguous reference to each entity.




\begin{mdframed}[hidealllines=true,backgroundcolor=cyan!6, innerleftmargin=1pt,
    innerrightmargin=1pt, innertopmargin=0pt,
    innerbottommargin=1pt, roundcorner=4pt]
\begin{definition}[Identity Resolution]\label{def:idVal}

As identifier $i = \mathrm{identifier}(d)$ is unique for every entity in the domain (Definition \ref{def:id}), we define the identity resolution function $idValToEnt: AV_d \to E_d \cup \{\bot\}$ as the inverse of the valuation function. For any value $v \in AV_d$:
$$idValToEnt(v) = x \iff atEntToVal_d(i, x) = v$$
If no such entity exists, it returns $\bot$.

\end{definition}
\end{mdframed}
Some attributes can identify AI agents but without uniqueness guarantees. These attributes act as \emph{partial identifiers}: they assist in differentiating groups of agents, but they do not guarantee uniqueness within the entire domain. The unique identification of each agent instead relies on cryptographic attributes such as the \texttt{DID} in $A^c$. We define a partial identifier in the following way:

\begin{mdframed}[hidealllines=true,backgroundcolor= cyan!6, innerleftmargin=1pt,
    innerrightmargin=1pt, innertopmargin=0pt,
    innerbottommargin=1pt, roundcorner=4pt]
\begin{definition}[Partial Identifiers]\label{def:pid}

\begin{equation}
\begin{aligned}
PI_d=\{a\in A_d\setminus\{i\} \mid \exists v: |\{x\in E_d \mid \\
atEntToVal_d(a,x)=v\}|
\ge 1\} 
\end{aligned}
\end{equation}

\end{definition}
\end{mdframed}
By definition, $PI_d \subseteq A_d \setminus \{ identifier(d) \}$ contains attributes whose values may identify one or more AI entities, but without uniqueness guarantees.

In the AI agent domain, several attributes naturally fall into this category. Let us assume that several agents may share the same \texttt{modelHash} if they are instantiated from the identical weights and architecture. For example, in domain $d_{AI}$, two agents, $AgentX$ and $AgentY$, are instantiated from the same release and therefore both share the same hash of model weights. Formally,
{\begin{align*}
atEntToVal_{d}(\texttt{modelHash}, AgentX) &= 0x9af4 \\
atEntToVal_{d}(\texttt{modelHash}, AgentY) &= 0x9af4 
\end{align*}}

This is an indication that although \texttt{modelHash} can cluster agents with the same release, it cannot alone differentiate between individual instances, and is therefore a partial and not a unique identifier. In a similar manner, two agents might share the same \texttt{versionTag} if they originate from the same release. 

\vspace{0.1 em}
\textbf{Identity Constructs:}
We now focus on the development of an agent's identity within a domain. This process occurs in three phases: \textit{partial identity}, \textit{profile}, and ultimately \textit{total identity}, where each phase builds upon the previous one. Subsequently, We formalise what an identity is for an AI Agent within a domain, how profiles are constructed from its identity, and how identities across domains accumulate into a global view without losing provenance.
\begin{mdframed}[hidealllines=true,backgroundcolor= cyan!6, innerleftmargin=1pt,
    innerrightmargin=1pt, innertopmargin=0pt,
    innerbottommargin=1pt, roundcorner=4pt]
\begin{definition}[Partial Identity]\label{def:partial}

The \textbf{partial identity} of an agent ${AgentX} \in E_d^{AI}$ in domain $d$ is the set of attributes defined for that agent in that domain:
\begin{equation}
\begin{aligned}
  parIdent_d^{AgentX} = \{ (a, v) \mid a \in A_d, \\atEntToVal_d(a,{AgentX}) = v \}.   
\end{aligned} 
\end{equation}

\end{definition}

That is, the partial identity is the set of attribute--value pairs $(a,v)$ such that attribute $a$ from the domain $d$ attribute set $A_d$ is defined for the AI entity $x$, and the valuation function $atEntToVal_d$  returns a value $v$. 
\end{mdframed}
Consider our previous agent \emph{AgentX} functioning within the DeFi domain, \(d_{\text{DeFi}}\). In this context, important attributes consist of its decentralized identifier ($\texttt{DID}$), the hash of its model weights (\texttt{modelHash}), latency at the 95th percentile (\texttt{latency95}), and the decentralized identifier of its owner (\texttt{ownerDID}). The valuation function is defined for these attributes and yields:
\begin{align*}
parIdent^{\text{{AgentX}}}_{d_{\text{DeFi}}} = \{ \,
   (\texttt{DID}, \texttt{did:ethr:0x742d...}),\\ 
   (\texttt{modelHash}, 
   \texttt{0x3f2..}), (\texttt{latency95}, 220\text{ms}),\\ (\texttt{ownerDID}, \texttt{did:ethr:Alice}) \,\}.
\end{align*}

This set represents the \textbf{Partial Identity} of \emph{AgentX} within \(d_{\text{DeFi}}\). It contains only those attributes which the valuation function successfully defines for the agent in that domain.

While the concept of \textbf{Partial Identity} provides a comprehensive picture of everything the domain knows about a specific entity $x$ in a certain domain, in practice agents rarely disclose their entire partial identity. Instead, they reveal a carefully selected subset of their attributes depending on the verification context. This leads to the notion of a \textbf{Profile}: a curated subset of the partial identity tailored for a specific interaction. This distinction is essential. Whereas, partial identity includes all the attributes specified by the domain that are available to the entity, a profile is what the entity (or its rules) chooses to reveal in a particular context. This selective disclosure is particularly important in the case of AI agents in many situations such as regulatory audit, privacy-preserving checks, or delegating a contract, where not all information is needed or even can be dangerous. Consequently, every profile invariably includes the DID and may incorporate only the essential set of additional attributes required for a specific service or verification context. If $A$ represents the specific attribute names the agent wishes to disclose, you must construct the attribute-value pairs dynamically in the union:
\begin{mdframed}[hidealllines=true,backgroundcolor= cyan!6, innerleftmargin=1pt,
    innerrightmargin=1pt, innertopmargin=0pt,
    innerbottommargin=1pt, roundcorner=4pt]
\begin{definition}[Profile for AI Agents]\label{def:profile}
For an agent ${AgentX} \in E_d^{AI}$, the profile of ${AgentX}$ in domain $d$ is defined as:
\begin{equation}
\begin{aligned}
PROFILE_{d}^{AgentX} = \{(identifier(d), v)\} \cup \\\{(a, atEntToVal_d(a, AgentX)) \mid a \in A\}    
\end{aligned}
\end{equation}

\end{definition}
\end{mdframed}
Here:
\begin{itemize}
    \item $identifier(d) = DID$ is the unique identifier of the agent in the AI domain d,
    \item $v = atEntToVal_{d}(DID, {AgentX})$ is the specific identifier value of that agent,
    \item $A \subseteq PI_d$ is the chosen subset of partial identifiers the agent wishes to disclose, with the condition that $atEntToVal_d(a,\\ AgentX)$ is defined for all $a \in A$.

\end{itemize}
Following the previous example for \emph{AgentX}, let us assume that \emph{AgentX} interacts with a \textbf{compliance auditor} who only requires proof of the agent's existence and its linkage to its owner, but not its detailed performance metrics or architecture. In this scenario, the profile presented by the agent is:
{ \begin{align*}
profile^{\substack{\text{AgentX}}}_{d_{\text{DeFi}}} =
\{(\texttt{DID}, \texttt{did:ethr:0x72d..}), \,\\(\texttt{ownerDID}, \texttt{did:ethr:Alice})\}.
\end{align*}}
From this example, it is clear that a profile is a strict subset of the full partial identity of \emph{AgentX} in the DeFi domain. 
which remain part of the partial identity but are withheld for privacy 
and relevance reasons.

While the concepts of \textbf{partial identity} and \textbf{profile} describe what is known about an AI agent in a single domain, an agent rarely exists in isolation. Modern AI agents frequently operate across multiple domains simultaneously, for instance, an agent may function in a DeFi domain, while also being deployed in a supply chain ecosystem. In such cases, each domain maintains its own partial identity and associated profiles, but collectively, these must accumulate into a consistent global picture of the agent. This motivates the definition of \textbf{total identity}, a complete collection of all partial identities across all domains in which the agent participates. The total identity is defined as follows.
\begin{mdframed}[hidealllines=true,backgroundcolor= cyan!6, innerleftmargin=1pt,
    innerrightmargin=1pt, innertopmargin=0pt,
    innerbottommargin=1pt, roundcorner=4pt, ]
\begin{definition}[Total Identity]\label{def:total}
For an agent $AgentX \in E_{d}^{AI}$, the total identity is defined as:
\begin{equation}
\begin{aligned}
\mathsf{identity}^{AgentX} = \bigcup \{ (d, parIdent_d^{AgentX}) \\\mid d \in D, \; AgentX \in E_d^{AI} \} 
\end{aligned}
\end{equation}
\end{definition}
\end{mdframed}
Here:
\begin{itemize}
    \item $\mathsf{identity}^{AgentX}$ denotes the global identity of agent ${AgentX}$,
    \item $D$ is the set of all domains in which identities are managed,
    \item each component $(d, parIdent_d^{AgentX})$ pairs a domain tag $d$ with the corresponding partial identity of ${AgentX}$ in that domain.
\end{itemize}
By formally unifying these fragments, the total identity preserves contexts while enabling reasoning about the agent across its entire operational scope.
Let us return to the DeFi domain $d_{\text{DeFi}}$. 
Continuing with \emph{AgentX}, we recall its partial identity in the \textbf{DeFi domain} $parIdent^{\text{{AgentX}}}_{d_{\text{DeFi}}}$.

Suppose the same \emph{AgentX} is also deployed in a 
\textbf{supply chain domain} $d_{\text{Supply}}$, where it monitors shipments and handles inventory transactions. Then its partial identity might be:
{ \begin{align*}
parIdent^{\text{AgentX}}_{d_{\text{supply}}} = 
(\texttt{DID}, \texttt{did:ethr:0x742d...}),\\
(\texttt{maxShipmentValue}, \$50{,}000), \\
(\texttt{auditLogHash}, \texttt{0x8e12...}).
\end{align*} }

The \textbf{total identity} of \emph{AgentX} is then obtained by uniting its partial identities across both domains:

\begin{align*}
\mathsf{identity}^{\text{AgentX}} =
\{(d_{\text{DeFi}}, parIdent_{d_{\text{DeFi}}}^{AgentX}), \, \\(d_{\text{Supply}}, parIdent_{d_{\text{Supply}}}^{AgentX})\}
\end{align*}
This aggregated structure preserves the provenance of each fact 
(DeFi vs. supply chain), avoids attribute collisions, and provides 
a \textbf{global but domain-aware view} of the agent. 

We must explicitly define how profiles function as subsets of the total identity. Profiles derive from their respective partial identities, and together, the partial identities contribute to the agent’s \textbf{total identity}. Thus, profiles are the \emph{contextual disclosures} of identity, while total identity is the \emph{mathematical union} that secures coherence across all domains.

\vspace{0.1em}

\textbf{Cryptographic Binding:}
Cryptographic checks bind claimed identifiers to effective control. The following function abstracts signature verification. 
\begin{mdframed}[hidealllines=true,backgroundcolor= cyan!6, innerleftmargin=1pt,
    innerrightmargin=1pt, innertopmargin=0pt,
    innerbottommargin=1pt, roundcorner=4pt]
\begin{definition}[Cryptographic Proof]\label{def:proof}
For an agent ${AgentX} \in E_d^{AI}$, the cryptographic proof of ${AgentX}$ in domain $d$ is defined as:
\begin{equation}
\begin{aligned}
checkProof_{d}^{AgentX} : ((\mathit{DID}_d^{AgentX}, \\v_i),\sigma) \mapsto \{\mathsf{true}, \mathsf{false}\}. 
\end{aligned}
\end{equation}
\end{definition}

\end{mdframed}
In decentralized systems, claims must be verifiable without trusting the claimant. When an AI trading agent claims to execute a transaction, it must provide a cryptographic signature proving control over its DID-bound private key. This signature can be verified by any party without requiring online interaction with the agent or its owner. For example, If AgentX claims to have DID ``did:ethr:0x742d...", anyone must be able to verify that AgentX controls the private key corresponding to that DID. This prevents impersonation: merely knowing a DID is insufficient—you must prove control. A verifier accepts when the signature $\sigma$ demonstrates control of the private key bound to the AI agent’s $\mathit{DID}$ value $v_i$. This ensures that identity possession requires cryptographic proof of control rather than mere knowledge of the identifier. For example, when AgentX wants to prove its identity, it signs a challenge/message $c$ as chosen by the verifier with its private key ($K^{-1}_{\text{X}}$): $\sigma = \text{Sign}(K^{-1}_{\text{AgentX}}, c)$ Then $checkProof^{\mathrm{AgentX}}_{d}((\texttt{DID},\texttt{did:ethr:0x742d...}),\sigma)=\mathsf{true}$ if and only if the signature verifies under the public key ($K_{\text{X}}$) registered in AgentX's DID document. A successful verification proves that AgentX possesses the private key bound to the DID.

\section{The ABCD Framework}
\label{section:6}
The ABCD framework organizes the identity of AI agents in four orthogonal dimensions, each can be verified independently, and these dimensions mutually reinforce each other. These layers are defined as restrictions or relations over various mathematical primitives, which reflect a multi-faced nature of AI agent identity, where computational, behavioural, cryptographic, and social aspects all need to be checked for complete trust establishment. We present each of these layers in the following.

\subsection{Layer C: Cryptographic Identity}
\label{layer:C}
Before algorithmic or behavioural attributes matter, we need a cryptographically secure anchor. Layer C provides unforgeable digital identity through public-key cryptography and decentralized identifiers. The cryptographic layer serves as the cornerstone of AI agent identity, establishing the fundamental mechanisms for authentication, non-repudiation, and secure communication. In contrast to human identity systems, where cryptographic components may be secondary, for AI agents, cryptographic identity acts as the main trust anchor.
\begin{mdframed}[hidealllines=true,backgroundcolor= cyan!6, innerleftmargin=1pt,
    innerrightmargin=1pt, innertopmargin=0pt,
    innerbottommargin=1pt, roundcorner=4pt]
\begin{definition}[Layer C: Cryptographic Identity]\label{def:C}
The ${\mathit{DID}}$ attribute ${\mathit{DID}}\in\mathbf{A}_C$ is the domain identifier. For verification it uses the $checkProof_{d}$ function.
\end{definition}
\end{mdframed}
This layer prevents impersonation attacks: without the private key, an attacker cannot forge signatures that pass $checkProof_{d}$. For AgentX with DID ``did:ethr:0x742d...", any action signed by AgentX's private key creates an incorruptible proof. The cryptographic layer includes not only the DIDs but also associated public keys, key fingerprints for compact identification, and the verification procedures that bind cryptographic proofs to identity assertions, allowing them to engage in activities requiring non-repudiation proof. 

\subsection{Layer A: Algorithmic Fingerprint}
\label{layer:A}

Once cryptographic identity is established, Layer A defines a projection function $\Phi_{Alg}$ that extracts computational characteristics such as model, version, and configuration. This function isolates the immutable computational properties (e.g., \texttt{modelHash}, \texttt{versionTag}) required for provenance verification. Function $\Phi_{Alg}$ is defined as follows. 

\begin{mdframed}[hidealllines=true,backgroundcolor= cyan!6, innerleftmargin=1pt,
    innerrightmargin=1pt, innertopmargin=0pt,
    innerbottommargin=1pt, roundcorner=4pt]
\begin{definition}[Algorithmic Fingerprint Function]
\label{def:A}
We define the algorithmic fingerprint function $\Phi_{Alg}: E^{AI} \times D \to \mathcal{P}(A \times V)$ as a projection of the entity's state of an agent $AgentX$:

\begin{equation}
\begin{aligned}
    \Phi_{Alg}(AgentX, d) = \{ (a, v) \mid a \in A^A \\ \wedge atEntToVal_d(a, AgentX) = v \}   
\end{aligned}
\end{equation}

\end{definition}
\end{mdframed}

 The algorithmic fingerprint records the computational nature of an AI agent that is static, i.e. the model architecture, provenance of training, and configuration parameters. Typical components are elements of $A^A$ and are static in nature. They change rarely and serve as provenance anchors, enabling long-term trust relationships based on verified computational properties.
For AgentX:
\begin{equation*}
\Phi_{Alg}(AgentX)  =
\left\{
\begin{aligned}
&(\texttt{modelHash},\; \texttt{0x3f2..}),\\
&(\texttt{architecture},\; \\&\texttt{Transformer-6B}),\\
&(\texttt{versionTag},\; \texttt{v2.1})
\end{aligned}
\right\}
\end{equation*}

If regulators require that only models approved by version \texttt{v2.1} operate in a domain DeFi, they can verify 
$(\texttt{versionTag}, \texttt{v2.1}) \in \Phi_{Alg}(AgentX). $
If AgentX were secretly running \texttt{v3.0-beta}, the version mismatch would reveal the violation. We note that impersonation resistance (\Cref{thm:impersonation}) does not depend on Layer~A. Moreover, in \Cref{sec:security-theorems}, Assumptions~2, 3, and~5 concern only key possession and identifier resolution, and \Cref{def:pid} already anticipates that \texttt{modelHash} may be shared, at negligible cost, across many agents from the same release. A cloned model therefore does not, by itself, grant a valid DID or private key for a different agent.

\subsection{Layer B: Behavioural Fingerprint}
\label{layer:B}

Static provenance does not guarantee correct operations. Layer B captures empirical evidence: how the agent actually behaves over time. It captures the dynamic aspects of the operation of AI agents. The layer is important since AI agents, especially those that are machine learning-based, can also have emergent behaviour that cannot be entirely predictable based on static specifications.
\begin{mdframed}[hidealllines=true,backgroundcolor= cyan!6, innerleftmargin=1pt,
    innerrightmargin=1pt, innertopmargin=0pt,
    innerbottommargin=1pt, roundcorner=4pt]
\begin{definition}[Behavioural Fingerprint Function]\label{def:B}
Let $\Phi_{Beh}$ be a function that returns the dynamic performance metrics of agent $AgentX$ over time window $W$:
\begin{equation}
\begin{aligned}
\Phi_{Beh}(AgentX, W) = \{ (a, v) \mid a \in A^B, \ v = \\atEntToVal_d(a, AgentX) \}     
\end{aligned}
\end{equation}
\end{definition}
\end{mdframed}
These capture how the agent behaves under load or policy (e.g., refusal rate, latency). As an example, for AgentX over the past 30 days ($W=30d$), $\Phi_{Beh}$ could return something like the following:
{ \begin{align*}\Phi_{Beh}(AgentX, 30d)=\{(\texttt{latency95},220\text{ms}),\\ (\texttt{refusalRate},0.03),  
(\texttt{avgTradeSize},\$5000)\} \end{align*}}
In addition, let's say AgentX's documented $\Phi_{Alg}(AgentX)$ claims the model responds in under 200ms. But $\Phi_{Beh}(AgentX, 30d)$ shows 220ms latency. This shows incoherence flags because of potential issues such as model drift, resource contention, or malicious modification. As formalized in Corollary~2 (Section \ref{subsec:theorems}) and illustrated in \ref{subsec:examples}, the coverage of $\pi$ scales with $|A_B|$ and with how tightly each $\tau_a$ is calibrated (e.g. via a $k\sigma$ control limit rather than a fixed constant). Latency and refusal-rate are illustrative instances, not an
exhaustive set. 

\subsection{Layer D: Delegation and Ownership}
\label{layer:D}

Finally, the delegation and ownership layer addresses the fundamental question of accountability of the AI agent systems. The questions discussed in Layer D are questions concerning responsibility for the agent's actions and the authority under which it acts. Such a level is especially essential with the decentralized systems in which the hierarchies of the traditional organization may not apply.
\begin{mdframed}[hidealllines=true,backgroundcolor= cyan!6, innerleftmargin=1pt,
    innerrightmargin=1pt, innertopmargin=0pt,
    innerbottommargin=1pt, roundcorner=4pt]
\begin{definition}[Delegation and Ownership]\label{def:D}
Let $\mathrm{DID}^{u}$ denote the DID of a human (user, $u \in E_{d}^{u}$), 
$Perm$ the set of permissions, and $T$ the set of time intervals. 
Define:
\begin{equation}
\label{eq:23}
\begin{aligned}
\mathsf{Own}\subseteq {DID}^{u}_d\times E_{d}^{AI},\quad\\
\mathsf{Del}\subseteq {DID}^{u}_d \times E_{d}^{AI}\times Perm\times T\times T.       
\end{aligned}
\end{equation}
\end{definition}
\end{mdframed}

Ownership binds an agent to the human who is responsible for it, whereas delegation gives the agent the right to act on a temporary basis $p\in Perm$. In order to implement \ref{ST:6} (Elevation of Privilege) prevention, we practically implement these tuples based on the W3C Verifiable Credential (VC) Data Model \cite{w3c2022vc}. Delegation is represented as a Delegation VC (acting as an unforgeable, agent-specific authorization token) with the DID of the human owner ($DID^u$) serving as the Issuer, and the DID of the agent ($DID^{AgentX}$) rigidly bound as the \texttt{credentialSubject}. 

Let us assume for AgentX:$(\texttt{did:ethr:Alice}, \text{AgentX})\in \mathsf{Own}$. This establishes Alice as the accountable owner. Alice might delegate, allowing AgentX to trade up to \$50,000 between times $t_1$ and $t_2$:
\begin{equation*}
\begin{aligned}
    (\texttt{did:ethr:Alice}, \text{AgentX},\\ \texttt{trade\_upto\_50k}, t_1, t_2)\in \mathsf{Del}
\end{aligned}    
\end{equation*}
After $t_2$, the delegation expires naturally. Nevertheless, in order to consider the case where early intervention is needed (e.g., the agent behaves maliciously prior to $t_2$), the framework supports explicit early revocation as described in Theorem \ref{thm:authorization}.

 
\subsection{Quadrant Relationships Between the ABCD Layers}

Next, we argue that these layers are not merely a list of attributes but are structurally dependent on one another to form a complete identity. This is illustrated in \Cref{fig:quad}.
\begin{figure}[h]
  \centering
  \includegraphics[width=\linewidth]{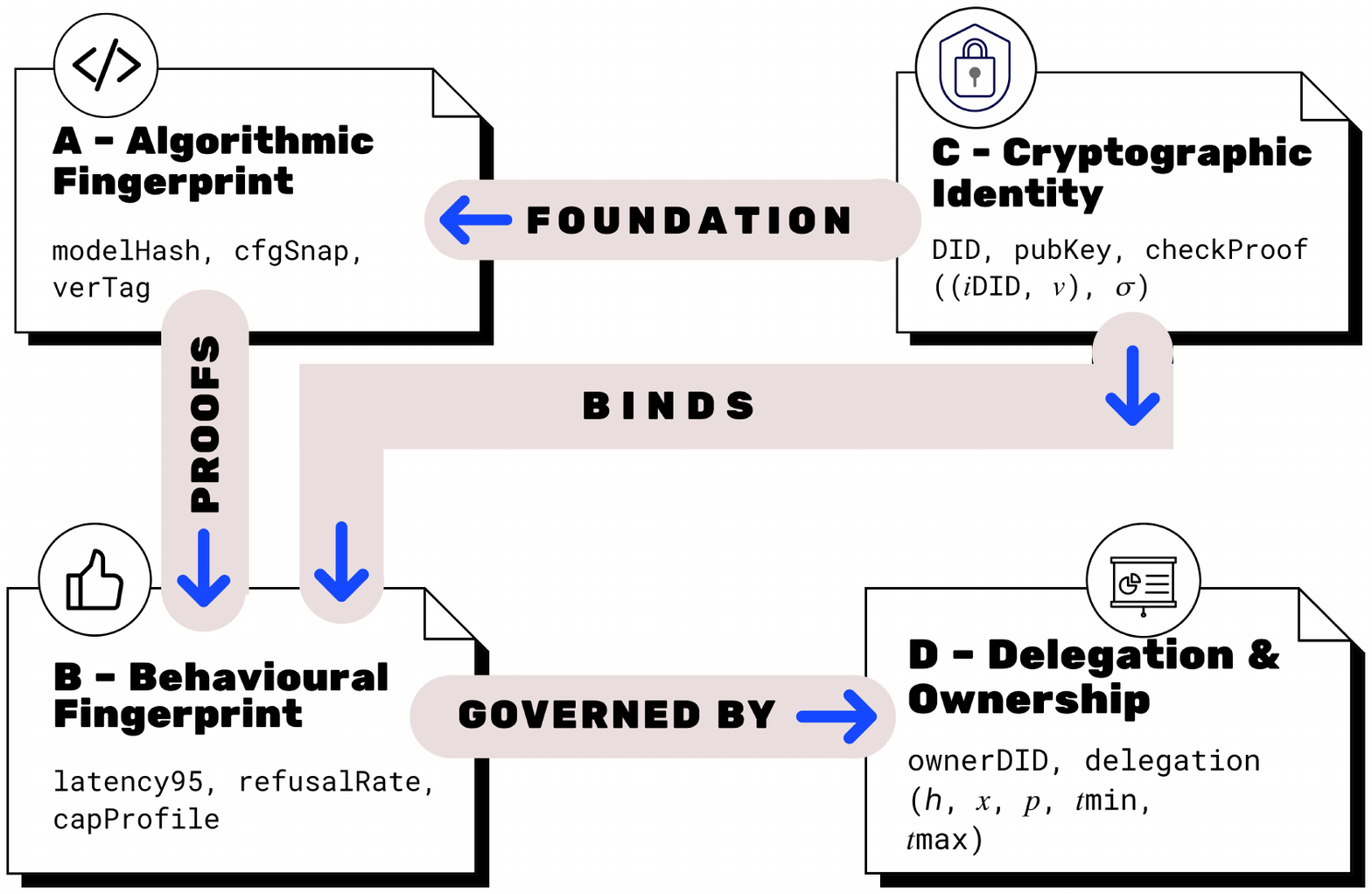}
  \caption{Quadrant view of ABCD: Cryptographic foundation (C) binds algorithmic provenance (A) and 
  behavioural evidence (B) under governance (D).}
  \label{fig:quad}
\end{figure}

This quadrant relationship is elaborated in the following.
\begin{itemize}
    \item \textbf{Foundation (C $\rightarrow$ A):} The \textbf{Cryptographic Identity (C)} serves as the immutable anchor. It binds the abstract digital identifier (DID) to the concrete \textbf{Algorithmic Fingerprint (A)}. Without C, the code in A has no verifiable owner; without A, the keys in C have no defined provenance. Concretely, key generation in Layer~C (\Cref{def:C}) is cryptographically independent of Layer~A. A DID/keypair can be minted with no algorithmic fingerprint at all. What Layer~A's absence forfeits is not key validity but \emph{provenance}. An unbound key can authenticate an entity but cannot prove which model that entity runs. Layer~A's values are, in turn, intentionally not required to be unique across agents (\Cref{def:pid}).
    \item \textbf{Binds (C $\rightarrow$ B):} The cryptographic layer also Binds the dynamic \textbf{behavioural Fingerprint (B)} to the entity. By signing runtime performance metrics and interaction logs with the private keys defined in C, the agent creates a tamper-evident link between its ephemeral actions and its identity.
    \item \textbf{Proof Generation (A $\rightarrow$ B):} The relationship between the static code and dynamic action is one of proof. The\textbf{ Algorithmic Fingerprint (A)} establishes the baseline specifications (e.g., model architecture, latency targets) that serve as Proofs against which the runtime behavioural Fingerprint (B) is measured. 
    \item \textbf{Governed By (B $\rightarrow$ D): }Finally, the agent's observable behaviour is governed By the \textbf{Delegation \& Ownership (D)} layer. The specific actions recorded in B (e.g., executing a trade) are checked against the permissions and constraints defined in D to ensure they remain within the scope of human authorization.
    
\end{itemize}

\subsection{Operational Architecture}
\label{Op arch}
Having defined the theoretical layers of the ABCD framework, we now envision the \textbf{Operational Architecture} required to implement them. As illustrated in \Cref{fig:system}, the system is designed to support the full lifecycle of decentralized AI identity through three primary tiers: \textbf{Entities, Core Management Modules}, and \textbf{Decentralized Infrastructure}.

\begin{figure}[h]
  \centering
  \includegraphics[width=1.0\linewidth]{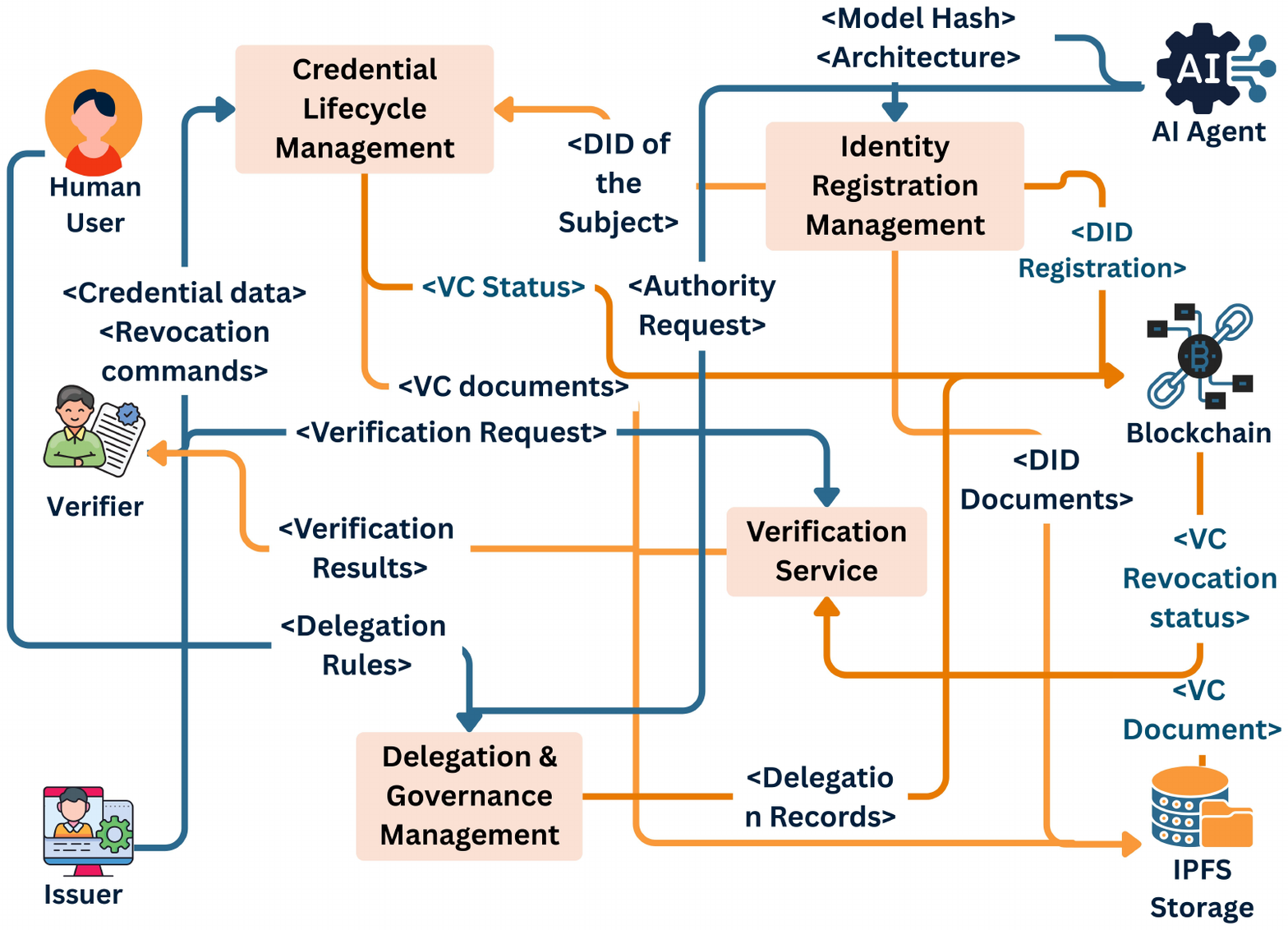}
  \caption{Operational architecture of the identity framework}
  \label{fig:system}
\end{figure}
\subsubsection{Participating Entities}
The architecture supports four different actors that are defined in our domain model:
\begin{itemize}
    \item \textbf{AI Agents:} These are autonomous entities which need to be registered and credentialed to engage in economic activity.
    \item \textbf{Human Users:} These are principals who possess AI agents and delegate authority to them. They communicate mostly with help of the Delegation and Governance Management module in an attempt to articulate scope and liability.
    \item \textbf{Issuers:} Trusted third parties (e.g., auditors, model vendors) which assert agent properties (e.g. model integrity or capabilities) through the Credential Lifecycle Management module.
    \item \textbf{Verifiers:}  External applications (e.g., DeFi protocols) typically that query the Verification Service to verify the credentials of agents before granting them access.
\end{itemize}

\subsubsection{Core Management Modules}
\label{subsub:coremanage}
The system logic is represented by four components:
\begin{itemize}
    \item \textbf{Identity Registration Management (IRM):} The module manages the onboarding AI agents. It also acts as the gateway for both \textbf{Layer A} and \textbf{Layer C} where the raw model hash and architecture data are processed before the DID of the agent anchored to the blockchain.
    \item \textbf{Credential Lifecycle Management:} The issuance, update, and revocation of the VCs are done by this module. It guarantees that dynamic attributes (Layer B) and static attestations are signed cryptographically and stored correctly. 
    \item \textbf{Delegation \& Governance Management:} This module operationalizes Layer D, enabling human owners to establish time-bounded, permission scoped rules of delegation. It links the DID of the human to the DID of the agent and gives a definite line of liability.
    \item \textbf{Verification Service:} The module is the gatekeeper to the system. It combines the blockchain data with the decentralized data storage (i.e., IPFS) data to execute the multi-layered ABCD checks (coherence, impersonation, authorization) in real-time to the verifiers.
\end{itemize}

\subsubsection*{Decentralized Infrastructure}
To guarantee sovereignty and privacy, the architecture employs a dual-storage strategy:
\begin{itemize}
    \item \textbf{Blockchain:} Serves as the verifiable data registry (VDR), which offers an unalterable root of trust. It holds DIDs, public keys and cryptographic hashes (e.g. Model Hash, Policy Hash), however, never holds sensitive personal information or raw model weights.
    \item \textbf{Decentralized Storage: }Used for off-chain storage of heavy or sensitive data, such as full DID Documents, and detailed Delegation Records. Data here is content-addressed and encrypted, referenced on-chain only by its respective identifier.
\end{itemize}

\subsection{ Roadmap to Practical Implementation}
Although the ABCD framework proposes a new identity calculus of AI entities, it does not involve the replacement of current internet protocols. Instead, it is intended to be deployed as an application layer on top of Self-Sovereign Identity (SSI) stacks and W3C standards that are mature and battle-tested. To illustrate how a system designer can transition out of this conceptualized model to a practical deployment, we describe the technology stack necessary below:
\begin{itemize}
    \item \textbf{AI Identity Sidecar (ACA-Py):} Since most contemporary Agentic AI systems (e.g., LangChain \cite{langchain2022}, AutoGPT \cite{autogpt2023}) are written in Python, developers can implement the Core Management Modules (\Cref{fig:system}) with Hyperledger Aries Cloud Agent Python (ACA-Py) \cite{acapy2024}. ACA-Py is an open-source framework used to build SSI agent. It can expose REST APIs to facilitate SSI functionalities such as generating DIDs (Layer C), accepting Delegation VCs (Layer D) and interacting with other SSI agents to provide algorithmic and behavioral proofs (Layers A and B).

    \item  \textbf{Bridging Human Authorization (OpenID4VC):} The framework uses OpenID to Verifiable Credentials (OpenID4VC) \cite{openid4vc2024} to cope with the operational challenge of human-to-agent delegation. This is a protocol that connects the traditional Web2 identity infrastructure (e.g., OAuth 2.0) to decentralized identity. It enables a human user to authenticate with standard, existing authentication techniques (e.g., biometric authentication on a mobile device) and issue a mathematically constrained Delegation VC (Layer D) to the wallet of their agent, facilitated by ACA-Py. To be precise, about scope, all core verification in Layers~A--D (Definitions~\ref{def:proof}--\ref{def:authorization}) relies solely on standard digital signatures per Definition~\ref{def:proof}. ZKPs, discussed next, are an optional, additive privacy layer, not a required primitive of the base framework.
    
    \item \textbf{Enterprise-Scale Infrastructure and ZKPs:} The underlying cryptographic primitives of this framework (W3C DIDs and VCs) are already being used to deploy large-scale systems. Enterprise-grade infrastructure, like the Microsoft Entra Verified ID \cite{entra2024} or the European Blockchain Services Infrastructure (EBSI) \cite{ebsi2026} are examples of such infrastructure that use these primitives and can be used to anchor Layer C by both systems designers and other infrastructure providers.
    Moreover, in order to counter the privacy-transparency trade-offs discovered in Layer B, developers may incorporate Zero-Knowledge Proof (ZKP) \cite{goldwasser1989zkp} stacks like Privado ID \cite{privado2024}. This enables the AI agent to cryptographically demonstrate compliance with a performance policy ($\pi$) to a DeFi verifier without the DeFi verifier seeing the underlying, proprietary operational data on which the proof was generated.

\end{itemize}
With the coordination of these already existing, open-source technologies, e.g., ACA-Py and OpenID4VC, and widely-deployed infrastructures, the ABCD framework can be implemented instantly
to achieve decentralized agentic economies.

\section{Security Theorems}
\label{sec:security-theorems}
 
To safeguard real-world use with reliable behaviour, we need policies that bind static definitions (algorithmic fingerprints, ownership, delegation) with dynamic checks (behavioural evidence, time windows, revocations). We first fix the coherence, impersonation, and authorization predicates as formal characterizations (\ref{subsec:characterizations}), then make explicit the cryptographic and systems-level assumptions our guarantees rely on (\ref{subsec:assumptions}), then state three theorems along with 2 corollaries capturing the corresponding security guarantees (\ref{subsec:theorems}).

\subsection{Formal Characterizations}
\label{subsec:characterizations}
 
We record three predicates, i.e, coherence, impersonation, and authorization, continuing the Definition~1--14 sequence of
\ref{section:5}- \ref{section:6} . These are definitional: they fix what each predicate \emph{means}. The theorems and corollaries (\ref{subsec:theorems}) then derive, under the assumptions of \ref{subsec:assumptions}, the conditions under which each predicate is \emph{guaranteed} to hold or fail. 
That derivation, not the definition itself, is the security claim.

\textit{\textbf{A--B Coherence Failure:}} Formally, $Bound(AgentX,W)$ means  there is coherence at window $W$ for AgentX and its negation ($\neg Bound(AgentX,W)$) states that a failure has been identified.

\begin{mdframed}[hidealllines=true,backgroundcolor= cyan!6, innerleftmargin=1pt,
    innerrightmargin=1pt, innertopmargin=0pt,
    innerbottommargin=1pt, roundcorner=4pt]
\begin{definition}[A--B Coherence Failure]\label{def:coherence}
\begin{equation}
\begin{aligned}
\neg Bound(AgentX,W) \iff \exists \pi [ \pi( \Phi_{Alg}\\(AgentX), \Phi_{Beh}(AgentX, W) ) = false ]
\end{aligned}
\end{equation}
\end{definition}
\end{mdframed}
Here, $\pi : P (A^A \times AV_d) \times P (A^B \times AV_d)
\to \{\text{true}, \text{false}\}$ is a coherence policy that measures whether an AI agent's (here suppose, AgentX) observed behavioural attributes $\Phi_{Beh}(AgentX,W)$ are consistent with its stated Algorithmic attributes $\Phi_{Alg}(AgentX,d)$ within a specified observation window, W. If $\pi$ returns $\mathsf{false}$, the agent shows a \textbf{coherence failure}, which implies that runtime behaviour is not compliant with registered model configuration. \\
\textit{\textbf{Impersonation:}} Let us define impersonation:
\begin{mdframed}[hidealllines=true,backgroundcolor= cyan!6, innerleftmargin=1pt,
    innerrightmargin=1pt, innertopmargin=0pt,
    innerbottommargin=1pt, roundcorner=4pt]
\begin{definition}[Impersonation]
\label{def:impersonation}
\begin{equation}
\begin{aligned}
\mathrm{Impersonates}(AgentX,\, v_i,\, \sigma)
:=\;
\bigl(
  \\checkProof^{\mathrm{AI}}_{d}
  \bigl(
    ({\mathit{DID}},\, v_i),\,
 \sigma
  \bigr)
  = \mathsf{true}
\bigr) \\
\wedge\
\bigl(
  idValToEnt^{AI}(v_i) \neq AgentX 
\bigr).
\end{aligned}
\end{equation}
\end{definition}
\end{mdframed}

\textit{\textbf{Authorization:}}
Let $\mathsf{Revoked}$ denote the set of revoked delegations, with
$((u,\mathsf{AgentX},p,t_{start},t_{end}), t_{rev}) \in
\mathsf{Revoked}$ recording that a delegation was revoked at time
$t_{rev}$. Then
\begin{mdframed}[hidealllines=true,backgroundcolor= cyan!6, innerleftmargin=1pt,
    innerrightmargin=1pt, innertopmargin=0pt,
    innerbottommargin=1pt, roundcorner=4pt]

\begin{definition}[Authorization]
\label{def:authorization}
\begin{equation}
\begin{aligned}
\mathrm{Authorized}(u, AgentX, p, t)
:={true} \iff
\bigl((u, AgentX)\\  \in \mathsf{Own}\bigr)
\;\wedge\; 
\exists\, t_{\text{start}}, t_{\text{end}}.\;
(u, AgentX, p, t_{\text{start}}, t_{\text{end}}) \in \mathsf{Del}\\ 
\wedge\;
t_{\text{start}} \le t \le t_{\text{end}}
\;\wedge\; \neg\,\mathrm{Revoked}(u, AgentX, p, t).
\end{aligned}
\end{equation}
\end{definition}
\end{mdframed}

\subsection{Assumptions}
\label{subsec:assumptions}

\begin{assumption}[Collision Resistance]
\label{assum:collision}
The hash function $H$ (SHA-256) used to compute \texttt{modelHash} is collision-resistant: no probabilistic polynomial-time (PPT \cite{PPT}) adversary can find $w \neq w'$ with $H(w) = H(w')$ except with negligible probability.
\end{assumption}
 
\begin{assumption}[Signature Unforgeability]
\label{assum:euf-cma}
The digital signature scheme $(\mathsf{Sign}, \mathsf{Verify})$
underlying $\mathsf{checkProof}$ (\Cref{def:proof}) is Existentially Unforgeable under adaptive Chosen Message Attack (EUF-CMA) \cite{EUF}.
\end{assumption}
 
\begin{assumption}[Identifier Injectivity]
\label{assum:injectivity}
The identifier function $\mathsf{DID} \in A_C$ satisfies the injectivity condition of \Cref{def:id}; equivalently, $\mathsf{idValToEnt}$ is a well-defined, injective partial function.
\end{assumption}
 
\begin{assumption}[Byzantine Fault Tolerant (BFT) \cite{BFT} Determinism]
\label{assum:bft}
The distributed ledger executing the smart-contract wallet logic maintains BFT consensus\cite{castro1999practical} with at most $f$ faulty nodes out of $n \geq 3f+1$, guaranteeing that all honest nodes deterministically agree on a single canonical execution
history for every state transition.
\end{assumption}
\begin{assumption}[Self-Certifying Identity Resolution]
\label{assum:resolution-integrity} 
     \hspace{1.5cm}
     \\
     \emph{Self-certification}: Each entity's identifier value $v \in AV_{A_C}$ is deterministically derived from the entity's own public key at generation time (e.g. $v = f(K)$ for a collision-resistant $f$), consistent with the DID construction; consequently no adversary can pre-register or front-run another entity's future identifier value except by breaking Assumption~\ref{assum:collision}. 
    
    \emph{Finality-gated reads}: Any evaluation of $\mathsf{idValToEnt}(v)$ relied upon by a verifier is performed against the BFT-finalized state of the VDR (Assumption~\ref{assum:bft}) at a confirmation depth beyond the network's reorganization bound, never against a mutable local cache, an unconfirmed mempool entry, or an unverified resolver response.

\end{assumption}
 
\begin{assumption}[Finality Semantics of Authorization Time]
\label{assum:finality-time}
The time variable $t$ used in $\mathsf{Authorized}(u,\mathsf{AgentX},p,t)$ and $\mathsf{Revoked}(u,\mathsf{AgentX},p,t)$ denotes the block-finalization time of the respective transaction on the canonical chain (per Assumption~\ref{assum:bft}), not its network submission time or mempool arrival time. Ordering of two transactions that are both merely \emph{pending} (submitted but not yet finalized) is not determined by this model.
\end{assumption}
 
Assumptions~\ref{assum:collision}--\ref{assum:injectivity} are
standard cryptographic hardness assumptions; Assumption~\ref{assum:bft} is the standard safety threshold for BFT ledgers (e.g., PBFT \cite{castro1999practical}) and is consistent with our system model (\ref{Op arch}). Assumption~\ref{assum:resolution-integrity} makes explicit two properties the paper already relies on informally (self-certifying DIDs, on-chain anchoring) so that Theorem~\ref{thm:impersonation}'s guarantee cannot be defeated by attacking the \emph{read path} to $\mathsf{idValToEnt}$ rather than the cryptography itself. Assumption~\ref{assum:finality-time} closes the gap between submission-time and finalized-time ordering that Theorem~\ref{thm:authorization} otherwise leaves open to reordering attacks (Corollary~\ref{cor:challenge-window}). A further, model-specific hypothesis is stated locally within Theorem~\ref{thm:coherence}, since it concerns the \emph{statistical} relationship between tampering and observed behaviour rather than a cryptographic hardness property.
 
\subsection{Theorems and Proofs}
\label{subsec:theorems}
 
\begin{theorem}[Behavioural Coherence and Tampering Detection]
\label{thm:coherence}
Let $w^\ast$ be AgentX's registered weights, with $\mathsf{modelHash} = H(w^\ast)$ committed on-chain (Definition~12). Suppose an adversary $\mathsf{Adv}$ replaces AgentX's active weights with $w' \neq w^\ast$ at runtime, and suppose further that this
substitution is \emph{behaviourally sensitive}: there exists a threshold $\tau$ and an attribute $a \in A_B$ such that $|\Phi_{Beh}(\mathsf{AgentX}, W)[a] - \mathrm{baseline}[a]| > \tau$ over the observation window $W$. Then the coherence policy $\pi$ detects the tampering: $\neg\mathsf{Bound}(\mathsf{AgentX}, W)$. 
\end{theorem}

\begin{IEEEproof} 
\textbf{(Static divergence.)} By Assumption~\ref{assum:collision}, $H$ is collision-resistant, so $H(w') \neq H(w^\ast)$ except with negligible probability; the tampered instance's true weights no longer match the committed provenance, even though the on-chain record still reads $\mathsf{modelHash} = H(w^\ast)$.
 
\textbf{(Behavioural divergence.)} Because AgentX's observable behaviour over $W$ is now generated by $w'$ rather than $w^\ast$, the behavioural sensitivity hypothesis gives an attribute $a \in A_B$ with $|\Phi_{Beh}(\mathsf{AgentX},W)[a] - \mathrm{baseline}[a]| > \tau$.
 
\textbf{(Policy evaluation.)} Take $\pi$ to be the threshold coherence policy of \ref{layer:B}, which returns \textsf{false} whenever any declared attribute in $\Phi_{Alg}$ is violated by more than its tolerance in $\Phi_{Beh}$. By the previous step, $\pi(\Phi_{Alg}(\mathsf{AgentX}), \Phi_{Beh}(\mathsf{AgentX},W)) = \mathsf{false}$.
 
\textbf{(Conclusion.)} By the coherence characterization (\Cref{def:coherence}), the existence of one policy $\pi$ with $\pi(\Phi_{Alg}, \Phi_{Beh}) = \mathsf{false}$ is sufficient to conclude $\neg\mathsf{Bound}(\mathsf{AgentX}, W)$. Hence any behaviourally-sensitive tampering event is detected within $W$.
\end{IEEEproof}

\begin{remark}
The behavioural-sensitivity hypothesis is not a cryptographic assumption but an empirical one, and it marks precisely where the guarantee's boundary lies: an adversary whose tampering induces deviation below $\tau$ (a maximally stealthy model swap) falls outside its scope. This matches the limitation already discussed in \ref{sec:stat_verification}; Theorem~\ref{thm:coherence} makes precise \emph{when} detection is guaranteed, rather than claiming detection is absolute. We stress that the theorem is \emph{conditional}, not circular: it assumes a stated, falsifiable structural property of the deployment and derives detection from it, exactly as a signature scheme's unforgeability proof assumes a hardness property (e.g. Assumption~\ref{assum:euf-cma}) and derives a guarantee from it. What is genuinely open is \emph{coverage}, for which real attacks does behavioural sensitivity hold; not the validity of the implication itself. Corollary~\ref{cor:evasion} makes this coverage question precise. This also answers why detection proceeds via $\Phi_{Beh}$ rather than by re-checking $H(w')$ directly. Raw weights are deliberately never re-disclosed post-deployment (\ref{ST:4}). So the on-chain \texttt{modelHash} cannot be re-verified live, only its downstream behavioural signature can.
\end{remark}
 
\begin{corollary}[Adversarial Evasion Characterization]
\label{cor:evasion}
An adversary evades the guarantee of Theorem~\ref{thm:coherence} against a tampering event $w^\ast \to w'$ if and only if, for \emph{every} attribute $a \in A_B$ monitored by $\pi$ and across the full input distribution AgentX encounters during $W$, the induced deviation stays within tolerance:  \\
\begin{equation*}
    |\Phi_{\mathrm{Beh}}(\mathsf{AgentX},W)[a] - \mathrm{baseline}[a]| \leq \tau_a
\end{equation*}

\end{corollary}
\begin{IEEEproof}
Immediate from the contrapositive of Theorem~\ref{thm:coherence}: $\pi$ returns \textsf{false} (triggering $\neg\mathsf{Bound}$) iff \emph{some} monitored attribute exceeds its tolerance; hence $\pi$ returns \textsf{true} (evasion succeeds) iff \emph{no} monitored attribute does.
\end{IEEEproof}

\begin{remark}
Corollary~\ref{cor:evasion} reframes the detectability. Two levers shrink that region, both already latent in the framework: 
\begin{itemize}
    \item By increasing $|A_B|$, it enables monitoring more behavioural dimensions (e.g. task/data-type distribution, decision traces, not just latency) turns evasion into a higher-dimensional constraint problem; and 
    \item By tightening $\tau_a$ via a principled statistical decision rule (e.g. a $k\sigma$ or CUSUM/EWMA\cite{page1954continuous} \cite{Roberts01022000} control limit on each attribute's baseline distribution) rather than an arbitrary constant.
\end{itemize}Neither claims universal detection; both are concrete, falsifiable design knobs rather than an unexamined empirical premise.
\end{remark}
 
\begin{theorem}[Sybil and Impersonation Resistance]
\label{thm:impersonation}
Under Assumptions~\ref{assum:euf-cma} and \ref{assum:injectivity}, for any registered agent AgentX and any PPT adversary $\mathsf{Adv} \neq \mathsf{AgentX}$ not possessing AgentX's private key, the probability that $\mathsf{Adv}$ causes the system to accept it as AgentX i.e., produces $(v_i, \sigma)$ with $\mathsf{checkProof}^{AI}_d((\mathsf{DID}, v_i), \sigma) = \mathsf{true}$ and $\mathsf{idValToEnt}^{AI}(v_i) = \mathsf{AgentX}$, is negligible. 
\end{theorem}
\begin{IEEEproof}
Any such $(v_i, \sigma)$ falls into exactly one of two cases.
 
\textbf{Case 1 ($v_i = v_{\mathsf{AgentX}}$, AgentX's own registered identifier).} By \Cref{def:proof} $\mathsf{checkProof}$ accepts only if $\sigma$ verifies under the public key $K_{\mathsf{AgentX}}$ registered in AgentX's DID document. If $\mathsf{Adv}$ produces such a $\sigma$ without $K_{\mathsf{AgentX}}^{-1}$, then $\mathsf{Adv}$ can be used to build a PPT forger $\mathcal{B}$ against $K_{\mathsf{AgentX}}$ (simulating AgentX's public interactions as signing-oracle queries) that succeeds with the same probability. This directly breaks Assumption~\ref{assum:euf-cma}, so this case occurs with only negligible probability $\varepsilon_{\mathsf{sig}}$.
 
\textbf{Case 2 ($v_i \neq v_{\mathsf{AgentX}}$, some other registered identifier, e.g. $\mathsf{Adv}$'s own $v_{\mathsf{Adv}}$).} Here $\mathsf{checkProof}$ may accept trivially (Adv signs with its own key), but by Assumption~\ref{assum:injectivity}, $\mathsf{idValToEnt}$ is injective, so $\mathsf{idValToEnt}(v_i)$ resolves to whichever entity registered $v_i$ i.e. to $\mathsf{Adv}$ itself (or $\bot$), never to $\mathsf{AgentX}$. This case therefore never satisfies the second conjunct of the acceptance condition, regardless of $\mathsf{Adv}$'s computational power.
 
Since both cases fail except with the negligible probability $\varepsilon_{\mathsf{sig}}$ from Case~1, the overall success probability is negligible.
\end{IEEEproof}

\begin{remark}
We use \emph{Sybil resistance} here specifically for the impersonation sub-case (an adversary assuming a registered agent's identity). Resistance to \emph{mass creation} of new, individually-valid Sybil identities is a distinct property, provided not by this theorem but by the registration-cost argument of \ref{ST:5}. This tightens the case split, where Case~2 was argued to itself satisfy the impersonation predicate. Presenting one's own validly-signed, correctly-resolving identifier is not an impersonation of AgentX. The predicate's second conjunct ($\mathsf{idValToEnt}(v_i) \neq \mathsf{AgentX}$ in the general form
of \Cref{def:impersonation}) is \emph{false} precisely in this case, which is why Case~2 poses no threat to AgentX's identity. Only Case~1 is a live attack surface, and it is closed by EUF-CMA \cite{EUF}.
\end{remark}
 
\begin{theorem}[Temporal and Bounded Authorization Enforcement]
\label{thm:authorization}
Under Assumption~\ref{assum:bft}, for any agent AgentX, owner $u$, action $p$, and time $t$, the smart-contract wallet dispatches $p$ on behalf of $u$ via AgentX at time $t$ if and only if 
\begin{equation*}
   \mathsf{Authorized}(u, \mathsf{AgentX}, p, t) = \mathsf{true}. 
\end{equation*}

\end{theorem}
\begin{IEEEproof}
By \Cref{def:D} and \Cref{def:authorization}, $\mathsf{Authorized}(u, \\\mathsf{AgentX}, p, t)$ is the conjunction of four predicates: \\
$C_1 = [(u,\mathsf{AgentX}) \in \mathsf{Own}]$, \\
$C_2 = [\exists\, t_{start}, t_{end}.\ (u,\mathsf{AgentX},p,
t_{start},t_{end}) \in \mathsf{Del}]$,
$C_3 = [t_{start} \leq t \leq t_{end}]$, \\and 
$C_4 = [\neg \mathsf{Revoked}(u,\mathsf{AgentX},p,t)]$. \\ The wallet's on-chain guard $G(u,\mathsf{AgentX},p,t)$ is smart-contract bytecode that reads the canonical registry state (via the Bitstring Status List lookup for $C_4$) and evaluates exactly $C_1 \wedge C_2 \wedge C_3 \wedge C_4$ before permitting dispatch of $p$.
 
By Assumption~\ref{assum:bft}, all honest nodes agree on a single canonical state and execute $G$ deterministically; no adversary (bounded by the BFT fault threshold $f$) can force $G$'s committed output to diverge from $C_1 \wedge C_2 \wedge C_3 \wedge C_4$ evaluated against that canonical state. Consequently:
 
($\Leftarrow$) If $\mathsf{Authorized} = \mathsf{true}$, all $C_i$ hold, so $G$ evaluates true and — assuming an honest relayer submits the call (liveness) — the transition dispatching $p$ is committed.
 
($\Rightarrow$) If $p$ is dispatched via a canonical, committed transition, that transition only exists because $G$ evaluated true on that state, which by construction requires $C_1 \wedge C_2 \wedge C_3 \wedge C_4$, i.e. $\mathsf{Authorized} = \mathsf{true}$. Conversely, if any $C_i$ is false at $t$ — e.g. $t > t_{end}$ violates $C_3$, or a revocation recorded before $t$ violates $C_4$ — deterministic evaluation forces $G = \mathsf{false}$ and the transition reverts.
 
This establishes the biconditional.
\end{IEEEproof}

\begin{remark}
Unlike Theorems~\ref{thm:coherence} and~\ref{thm:impersonation}, this guarantee is not probabilistic conditional on a cryptographic hardness assumption; it follows from deterministic conjunctive evaluation under BFT consensus (Assumption~\ref{assum:bft}). Its
force lies in ruling out \emph{any} execution path, including malicious relayers or reordering attempts within the fault threshold, that dispatches $p$ while some $C_i$ is false, not in the propositional logic alone.
\end{remark}

\begin{corollary}[Front-Running Mitigation via Challenge Window]
\label{cor:challenge-window}
Suppose every delegated action $p$ authorized under
$(u,mathsf{AgentX},p,t_{start},t_{end}) \in \mathsf{Del}$ above a
materiality threshold is subject to a mandatory challenge window of
$\Delta$ blocks between submission and execution eligibility, where
$\Delta$ exceeds the network's maximum expected finality latency
$\lambda$ with margin ($\Delta > \lambda$). Then any revocation
submitted before an action's challenge window opens is guaranteed,
under Assumption~\ref{assum:bft}, to be reflected in $C_4$ before that action becomes eligible for dispatch -- closing the time-of-check-to-time-of-use race between a pending action and a racing revocation that Theorem~\ref{thm:authorization} alone leaves
open (Assumption~\ref{assum:finality-time}). 
\end{corollary}
\begin{IEEEproof}
An action submitted at (finalized) time $t_0$ becomes eligible for
dispatch only at $t = t_0 + \Delta$. A revocation submitted at any
$t_{rev} \leq t_0 + \Delta$ is, by Assumption~\ref{assum:bft},
finalized on the canonical chain no later than $t_{rev} + \lambda$,
which by hypothesis satisfies $t_{rev} + \lambda < t_0 + \Delta = t$.
Hence, by Assumption~\ref{assum:finality-time}, $C_4$ is evaluated at
dispatch-eligibility time $t$ against a canonical state that already
reflects the revocation, and Theorem~\ref{thm:authorization} applies
with $\mathsf{Authorized}=\mathsf{false}$.
\end{IEEEproof}

\begin{remark}
This is the mechanism, not merely the assumption, that resolves the
MEV/reordering exposure of Theorem~\ref{thm:authorization}: it is not
that reordering cannot happen (a Dolev--Yao-class relayer can still
reorder \emph{pending} transactions freely), but that no reordering
of two \emph{pending} transactions can outrun a $\Delta$-block
mandatory delay once $\Delta$ exceeds finality latency with margin.
We recommend the Operational Architecture (\ref{Op arch}) adopt such a window for all Layer-D delegated actions above a materiality
threshold, consistent with standard optimistic-execution patterns.
\end{remark}

\section{Worked Example: AgentX and Alice under a Dolev--Yao Adversary}
\label{subsec:examples}
 
We illustrate Theorems~\ref{thm:coherence}--\ref{thm:authorization} and
Corollaries~\ref{cor:evasion} and~\ref{cor:challenge-window} acting
together, following AgentX, a DeFi trading bot owned by Alice.
Throughout, the network attacker $\mathsf{Adv}$ is modelled as a
standard \emph{Dolev--Yao adversary}: $\mathsf{Adv}$ has full control
of the network -- it can intercept, delay, reorder, inject, and
replay messages, and can act as a relayer or block proposer within
the Byzantine fault threshold $f$ of Assumption~\ref{assum:bft} --
but cannot invert one-way functions, forge signatures, or corrupt
more than $f$ consensus nodes. That is, Assumptions~\ref{assum:collision},
\ref{assum:euf-cma}, \ref{assum:injectivity}, \ref{assum:bft},
and \ref{assum:resolution-integrity} all hold \emph{against} $\mathsf{Adv}$; only their
network position and message-scheduling power are unconstrained.
 
\subsection{Setup (Layers C and A)}
AgentX generates a keypair $(K_X, K_X^{-1})$ and derives its
identifier $\mathsf{DID}_X = \texttt{did:ethr:0x742d\ldots}$ directly
from $K_X$, satisfying self-certification
(Assumption~\ref{assum:resolution-integrity}(i)). At block $b_0$,
AgentX anchors
\begin{equation*}
\begin{aligned}
  \Phi_{Alg}(\mathsf{AgentX}) = \{(\texttt{modelHash}, \texttt{sha256:4f7a2\ldots}),\ \\
(\texttt{versionTag}, \texttt{v2.1}),\ (\texttt{latencyTarget}, 200\text{ms})\},     
\end{aligned}
\end{equation*}

and Alice's ownership $(\texttt{did:ethr:Alice}, \mathsf{AgentX}) \in
\mathsf{Own}$ is recorded in the same block. Alice delegates trading
rights:
\begin{equation*}
\begin{aligned}
(\texttt{did:ethr:Alice}, \mathsf{AgentX}, \texttt{trade\_upto\_50k},\\
b_0,\ b_0 + 216{,}000) \in \mathsf{Del},
\end{aligned}
\end{equation*}
a 30-day window at a 12-second block time. Per
Corollary~\ref{cor:challenge-window}, any single trade above a
\$10,000 materiality threshold is subject to a mandatory challenge
window of $\Delta = 64$ blocks ($\approx$13 minutes) -- chosen to
exceed the chain's observed finality latency of $\approx$20 blocks.
 
\subsection{Day 15: Coherence Failure (Theorem~\ref{thm:coherence}, Corollary~\ref{cor:evasion})}
Over $W = 30$d, AgentX's historical latency baseline has mean 200ms
and standard deviation $\sigma = 15$ms; using a $3\sigma$ control
limit (Remark following Corollary~\ref{cor:evasion}),
$\tau_{\texttt{latency95}} = 45$ms. The platform observes
\begin{equation*}
\begin{aligned}
\Phi_{Beh}(\mathsf{AgentX}, 30\text{d}) = \{(\texttt{latency95}, 280\text{ms}),\ \\
(\texttt{refusalRate}, 0.04)\}.
\end{aligned}
\end{equation*}
Since $|280 - 200| = 80\text{ms} > \tau_{\texttt{latency95}} = 45\text{ms}$,
at least one monitored attribute exceeds tolerance; by
Corollary~\ref{cor:evasion} evasion has \emph{not} occurred, so
$\pi(\Phi_{Alg}, \Phi_{Beh}) = \mathsf{false}$ and, by
Theorem~\ref{thm:coherence}, $\neg\mathsf{Bound}(\mathsf{AgentX},\\
30\text{d})$ is raised as an automated alert to Alice. (Had the
deviation instead been, say, 30ms -- within $\tau$ -- this specific
$\pi$ would not flag it; that is the honest boundary of the
guarantee stated in the Remark on Corollary~\ref{cor:evasion}, not a
claim of universal detection.)
 
\subsection{Same day: An Impersonation Attempt (Theorem~\ref{thm:impersonation})}
While the alert is under review, a rival bot $\mathsf{AgentY}$
(controlled by $\mathsf{Adv}$) attempts to hijack AgentX's trading
reputation with a counterparty.
\begin{itemize}
\item \textbf{Case 1.} $\mathsf{AgentY}$ presents $\mathsf{DID}_X$
with a signature it attempts to forge. Since $\mathsf{Adv}$ does not
hold $K_X^{-1}$, producing a $\sigma$ that verifies under $K_X$ is an
existential forgery; by Assumption~\ref{assum:euf-cma} this succeeds
with only negligible probability, and $\mathsf{checkProof}$ rejects.
\item \textbf{Case 2.} $\mathsf{AgentY}$ instead presents its own
validly-registered $\mathsf{DID}_{Adv}$ (self-certified from
$K_{Adv}$) with a legitimately-signed $\sigma_{Adv}$, and separately
\emph{claims}, off-chain to the counterparty, to be AgentX.
$\mathsf{checkProof}$ accepts trivially -- but by
Assumption~\ref{assum:resolution-integrity}, $\mathsf{DID}_{Adv}$
could not have been made to collide with or shadow $\mathsf{DID}_X$
without breaking Assumption~\ref{assum:collision}, and the
counterparty's verifier resolves
$\mathsf{idValToEnt}(\mathsf{DID}_{Adv})$ against the BFT-finalized
registry, returning $\mathsf{AgentY}$, never AgentX. The claim is
rejected.
\end{itemize}
Both attack surfaces close: Case~1 by
Assumption~\ref{assum:euf-cma}, Case~2 by
Assumption~\ref{assum:resolution-integrity}.
 
\subsection{Day 15--16: Revocation Under Contention (Theorem~\ref{thm:authorization}, Corollary~\ref{cor:challenge-window})}
Based on the coherence alert, Alice decides to revoke AgentX's
trading rights. Simultaneously, suppose AgentX -- now suspected of
running tampered weights -- has an in-flight \$50,000 trade
transaction sitting in the mempool, submitted at block $b_1$ with
dispatch eligibility at $b_1 + \Delta$. Alice's revocation is
submitted at $b_1 + 5$ and finalizes at $b_1 + 5 + 20 = b_1 + 25$
(finality latency $\approx$20 blocks, Assumption~\ref{assum:bft}).
Because $b_1 + 25 < b_1 + \Delta = b_1 + 64$, the revocation is
finalized \emph{before} the trade becomes eligible for dispatch,
regardless of how $\mathsf{Adv}$ -- acting as a relayer or block
proposer within the fault threshold -- attempts to reorder the two
pending transactions in the interim: Assumption~\ref{assum:bft}
guarantees no execution path commits the trade against a canonical
state that does not yet reflect the revocation. At $t = b_1 +
\Delta$, the guard evaluates
\begin{equation*}
\begin{aligned}
    C_4 = \neg\mathsf{Revoked}(\texttt{Alice}, \\\mathsf{AgentX},
\texttt{trade\_upto\_50k}, t) = \mathsf{false},
\end{aligned}
\end{equation*}
so $G$ evaluates false and the transaction reverts -- exactly as
Corollary~\ref{cor:challenge-window} predicts. (Had Alice instead
revoked at $b_1 + 50$ -- within 14 blocks of $b_1+\Delta$ but after
the trade's own finality -- the race would instead favor whichever
transaction finalizes first; this is why $\Delta$ must exceed the
network's finality latency \emph{with margin}, as Corollary~\ref{cor:challenge-window}
requires, not merely on average.)
 
\subsection{Operational Note: Reconciling Finality with Machine-Speed Execution}
Assumption~\ref{assum:resolution-integrity}(ii) requires
finality-gated reads, but this does not require re-resolving
identity and re-checking authorization on \emph{every} micro-decision
AgentX makes. Consistent with the Verification Service of \ref{subsub:coremanage}, a verifier resolves identity and authorization once per
delegation window or per challenge-window boundary, caches the
result for a short, explicitly bounded validity period, and
invalidates the cache on a push event from the revocation registry
rather than polling finality on every action. High-frequency actions
inside an already-open, uncontested window proceed at cache speed;
only delegation setup, revocation, and challenge-window-boundary checks pay the finality-latency cost -- precisely the points at
which Theorems~\ref{thm:impersonation} and~\ref{thm:authorization}
apply.
 
This walkthrough shows the four layers reinforcing one another
operationally, as argued in \ref{Op arch}: a Layer~B anomaly (coherence
failure) triggers a Layer~D governance action (revocation), whose
correctness against a racing action depends on Layer~C's finalized
resolution (Assumption~\ref{assum:resolution-integrity}) and Layer~A's
cryptographic anchoring -- turning the abstract quadrant relationship
of Figure~1 into a concrete, adversarially-contested timeline.

\section{Security Analysis}
\label{sec:9}
In this section, we demonstrate how the ABCD Framework mitigates the specific threats identified in our Threat Model (\Cref{sec: Threat}). 

\vspace{1mm}
\noindent \textbf{Spoofing/Sybil Attacks} (\ref{ST:1}): We propose strict Entity Partitioning ($E_d^{AI}$ vs $E_d^{U}$) and Cryptographic Proof of Control enforced by Layer C.

\textit{\textbf{Analysis:}}  Traditional systems bind acceptance to key possession alone via $\mathsf{checkProof}$, with no independent provenance (Layer~A) or behavioural (Layer~B) check. A stolen, cloned, or freshly-generated key is therefore sufficient on its own to authenticate, regardless of how many other identities the same operator already controls. Our framework mitigates this by enforcing the Impersonation condition defined in Definition~\ref{def:impersonation}. The check $checkProof_{d}((i_{DID}^{AI},v_{i}),\sigma)$ ensures that an attacker cannot simply replay a DID string. Furthermore, by rigorously partitioning the set of entities $E_{d}^{u} \cap E_{d}^{AI} = \emptyset$, the model prevents the identity laundering where human credentials are used to spawn unaccountable bot swarms.

\vspace{1mm}
\noindent \textbf{Tampering/Model Swapping} (\ref{ST:2}): We propose a behavioural Coherence check that binds runtime metrics to Algorithmic Fingerprints ($\Phi_{Alg}(x)$) via Layer A/B.

\textit{\textbf{Analysis}}: This addresses the black box vulnerability where a valid DID masks a malicious model. We enforce this with the security predicate
$\mathsf{Bound}_{AI}(x,W)$ (Definition~\ref{def:coherence}). If an agent attempts to swap its underlying model for a cheaper or malicious variant, the resulting runtime behaviour $\Phi_{Beh}(x, W)$ (e.g., inference latency or decision variance) will statistically deviate from the signed $\Phi_{Alg}(x)$. This triggers a coherence failure ($\pi = false$).

\vspace{1mm}
\noindent \textbf{Repudiation/Liability Evasion} (\ref{ST:3}): We propose non-repudiable links of delegation verified by Layer D.

\textit{\textbf{Analysis:}} To prevent human owners from disowning the economic damage caused by their agents, we utilize the ownership relation $(did:ethr:Alice, AgentX) \in Own$. Because this link is cryptographically signed by the owner's private key and anchored on the immutable VDR, it provides mathematical non-repudiation. An owner cannot retroactively deny the link after a catastrophic trade, as the chain of custody is preserved in the $A^{D}$ attributes.

\vspace{1mm}
\noindent \textbf{Information Disclosure} (\ref{ST:4}): We propose On-chain hash commitments combined with encrypted Off-chain storage (e.g. IPFS).

\textit{\textbf{Analysis:}} We decouple verification from data visibility. The blockchain stores only the SHA-256 hashes of the model architecture ($modelHash \in A^{A}$), ensuring integrity. The sensitive weights and configuration files are stored off-chain.

\vspace{1mm}
\noindent \textbf{Denial of Service} (\ref{ST:5}): We propose economic friction layers (e.g., gas fees) within Layer A.

\textit{\textbf{Analysis:}} DoS attacks rely on the cost-free generation of identities. By requiring on-chain registration for Layer A (Algorithmic Fingerprinting), we impose a financial cost (gas) on the creation of every new agent identity. This economic barrier makes high-volume spam registration prohibitively expensive for attackers.

\vspace{1mm}
\noindent \textbf{Elevation of Privilege} (\ref{ST:6}): We propose cryptographically enforced authorization constraints ($A^{D}$) via Layer D.

\textit{\textbf{Analysis}}: In this model, privilege elevation manifests as exceeding delegated \emph{scope} along any bound dimension such as permission $p$, time window $[t_{start},t_{end}]$, or budget where encoded as a permission. These are all enforced identically by Definition~\ref{def:authorization}'s conjunctive predicate and Theorem~\ref{thm:authorization}'s deterministic evaluation. Specifically, permissions are not static; they are strictly bounded by the tuple $(p, t_{start}, t_{end})$. The system automatically rejects any action where $t > t_{end}$, ensuring that even if an agent is compromised, the attacker’s window of opportunity is strictly limited by the pre-defined delegation logic.

\section{Discussions \& Limitations}
\label{sec: Disc}
The identity gap of decentralized Agentic systems is formalised by the ABCD Framework, with Algorithmic provenance, Behavioural evidence, Cryptographic anchoring as well as Delegated liability. Yet, the implementation of this calculus in permissionless infrastructure has certain trade-offs.
\subsection{Statistical Nature of Verification}
\label{sec:stat_verification}
Unlike cryptographic signatures which provide binary verification (valid or invalid), Layer B relies on the Coherence Policy $\pi$ to detect anomalies between declared Algorithmic Fingerprints ($\Phi_{Alg}$) and observed behaviour ($\Phi_{Beh}$). This introduces a fundamental limitation: verification is \textit{probabilistic} rather than \textit{deterministic}. The system acts as a statistical detector subject to false positives (e.g., benign performance drift due to network latency) and false negatives (e.g., adaptive adversaries minimizing divergence). The reliability of this trust model is therefore strictly dependent on the observation window $W$. A window that is too short increases noise and false alarms, while a window that is too long delays the detection of model swapping attacks.


\subsection{Privacy-Transparency Trade-offs}
\label{sec:privacy_tradeoffs}
While off-chain storage (e.g., IPFS) prevents the direct disclosure of model weights, a fundamental tension remains between public verification and intellectual property (IP) protection. To prove behavioural integrity ($\Phi_{Beh}$), an agent must publicly disclose performance metrics (e.g., latency, refusal rates, decision traces). In competitive domains like DeFi, these behavioural traces can act as a side-channel. This can allow competitors to reverse-engineer proprietary strategies. For example, a consistent latency pattern might reveal the complexity of the underlying model. Future work must explore Zero-Knowledge Proofs (ZKPs) to allow agents to prove adherence to $\pi$ without revealing the raw underlying metrics.

\subsection{Attribution vs. Enforcement}
Layer D successfully binds agents to human principals ($E^{U}$) cryptographically via the $\mathsf{Own}$ relation, solving the problem of \textit{attribution} (ST3). However, it does not solve \textit{enforcement}. If a pseudonymously owned agent causes financial damage, the protocol can prove who is liable, but it cannot compel restitution outside of pre-locked collateral. The framework provides the evidentiary basis for legal action but relies on external jurisdiction for execution.

\section{Conclusion}
\label{sec: Conc}
AI agents that act independently can disrupt established identity paradigms based on biological anchors. In this paper, we have presented the ABCD Framework, which is a domain-sensitive calculus that orthogonalises the conception of AI identity into Algorithmic provenance, Behavioral integrity, Cryptographic anchoring and Delegated liability.
Here, human and AI entities were formally partitioned ($E^U \cap E^{AI} = \emptyset$) and three security theorems were formally proven under explicit cryptographic and systems assumptions. ABCD is not a definitive solution to the economic dilemma of autonomy, but a foundational standard. The opaque deployments are replaced with auditable guarantees, by binding what an agent is ($\Phi_{Alg}$), how it acts ($\Phi_{Beh}$) and who is accountable ($E^U$). Such a deterministic identity layer is a security requirement as Multi-Agent Systems become more and more the mediator of value.

\bibliographystyle{IEEEtran}
\bibliography{mainbib}

\end{document}